\documentclass[smallcondensed]{svjour3}
\usepackage{tikz}
\usetikzlibrary{arrows}
\usetikzlibrary{cd}
\usepackage{mathtools}
\usepackage{hyperref}
\usepackage{newtxtext}
\usepackage{newtxmath}
\usepackage{subfigure}
\usepackage{xcolor}
\usepackage{comment}
\usetikzlibrary{decorations.markings}
\usepackage{bbm}
\usepackage{ifthen}

\newcommand\RR{\mathbb{R}}
\newcommand\CC{\mathbb{C}}
\newcommand\ZZ{\mathbb{Z}}

\newcommand\cG{\mathcal{G}}

\newcommand\bfB{\mathbf{B}}

\newcommand\bfa{\mathbf{a}}
\newcommand\bfb{\mathbf{b}}
\newcommand\bfe{\mathbf{e}}

\newcommand\bfx{\mathbf{x}}

\newcommand\graf{\Gamma}

\renewcommand\bar{\overline}
\renewcommand\tilde{\widetilde}

\DeclareMathOperator\st{Star}
\DeclareMathOperator\val{val}

\def\mH {\mathcal{H}}
\def\bone {\mathbbm{1}}

\begin{document}
\title{Geometric presentations of braid groups for particles on a graph}
\author{Byung Hee An \and Tomasz Maciazek}
\institute{Byung Hee An
\at Department of Mathematics Education, Kyungpook National University, Daegu, South Korea \\\email{anbyhee@knu.ac.kr}
\and
Tomasz Maciazek
              \at School of Mathematics, University of Bristol, Fry Building, Woodland Road, Bristol BS8 1UG, UK
           \and
Center for Theoretical Physics, Polish Academy of Sciences, Al. Lotnik\'ow 32/46, 02-668 Warszawa, Poland
}
\date{}
\maketitle

\begin{abstract}
We study geometric presentations of braid groups for particles that are constrained to move on a graph, i.e. a network consisting of nodes and edges. Our proposed set of generators consists of exchanges of pairs of particles on junctions of the graph and of certain circular moves where one particle travels around a simple cycle of the graph. We point out that so defined generators often do not satisfy the braiding relation known from 2D physics. We accomplish a full description of relations between the generators for star graphs where we derive certain quasi-braiding relations. We also describe how graph braid groups depend on the (graph-theoretic) connectivity of the graph. This is done in terms of quotients of graph braid groups where one-particle moves are put to identity. In particular, we show that for $3$-connected planar graphs such a quotient reconstructs the well-known planar braid group. For $2$-connected graphs this approach leads to generalisations of the Yang-Baxter equation. Our results are of particular relevance for the study of non-abelian anyons on networks showing new possibilities for non-abelian quantum statistics on graphs.
\end{abstract}

\section{Introduction}
The study of non-abelian quantum statistics is currently at the forefront of research concerning quantum computers \cite{qcomp-reviwew}, the fractional quantum Hall effect \cite{stern} and superconductivity \cite{Wilczek}. Exchange of (quasi-)particles obeying non-abelian quantum statistics results with a unitary transformation of the multi-component wave function of the considered quantum system. If the considered particles are constrained to move in $\RR^2$, this means that such a quantum system features a unitary representation of the braid group. It is however possible to generalise the idea of braiding to situations where particles are confined in a space which is different than $\RR^2$. In recent years there has been much interest in the study of non-abelian statistics and non-abelian anyons on spaces that have the topology of a network (also called a graph or a one-dimensional $CW$-complex). Realising anyons on a graph is a particularly promising concept in the context of topological quantum computing where computations are realised by adiabatic manipulation of anyonic quasi-particles. Such a precise control of anyons is believed to be achievable most easily on graphs. One of the leading proposals in this area utilises the exchange of Majorana fermions on networks consisting of coupled semiconducting nanowires \cite{alicea,majorana-review}. 

Following early seminal works on quantum statistics in $2D$ \cite{LM,Souriau}, in the mathematical description of anyons placed in topological space $X$ one considers the following configuration space $C_n(X)$.
\begin{equation}
C_n(X):=\left(X^n-\Delta_n\right)/S_n,
\end{equation} 
where $\Delta_n:=\left\{\left(x_1,\hdots,x_n\right):\ x_i\neq x_j{\mathrm{\ for\ all\ }} j\neq i\right\}$ is called the diagonal of $X^n$ and $S_n$ is the permutation group that acts in $X^n$ by permuting particles. The fundamental group of $C_n(X)$ is called the braid group and will be denoted by $\bfB_n(X)$. Importantly, elements of $\bfB_n(X)$ represent exchanges of particles. This is due to the fact that any loop in $C_n(X)$ can be viewed as a process where particles started in some initial configuration and returned to it modulo a permutation of particles. Note that removing the diagonal means that no two particles are allowed to occupy the same position in $X$. In such a formulation of quantum mechanics, wave functions are continuous functions $\Psi:\ C_n(X)\to\CC^d$. We allow the wave function to have $d>1$ components accounting for particle's internal degrees of freedom. Exchange of particles gives rise to a unitary transformation of the wave function $\Psi\to U\Psi$ where $U$ is a unitary $d\times d$ matrix. Such a transformation is required to be topological, i.e. it is assumed that operator $U$ depends only on the homotopy class of the loop describing particles' exchange. More formally, quantisations on $C_n(X)$ are in a one-to-one correspondence with conjugacy classes of unitary representations of $\bfB_n(X)$. Standard settings include $X=\RR^3$ or $X=\RR^2$. In the first case we have $\bfB_n(\RR^3)=S_n$ and unitary representations of $S_n$ give rise to so-called parafermions. On the other hand, unitary representations of the planar braid group $\bfB_n(\RR^2)$ give rise to non-abelian anyons relating to the fractional quantum Hall effect, tensor categories utilised in quantum computing \cite{freedman-qc} and the field-theoretic description of quantum statistics \cite{Wilczek,field}. In our paper we focus on $X=\Gamma$, a graph. 

Despite the notable interest in non-abelian statistics and anyons on graphs, relatively little is known about the most fundamental object that captures all information about particles' exchange, i.e. graph braid groups. Notably, there exists no explicit description of a universal set of generators and relations for graph braid groups in terms of geometric moves of particles on a given graph. Some sets of universal generators have been derived in terms of the discrete Morse theory \cite{FSbraid,FS12} or a recursive construction of graph configuration spaces \cite{kurlin}. In this paper, we prove that graph braid groups are universally generated by certain particle moves that have the following intuitive description: i) exchanges of pairs of particles on junctions of the considered graph and ii) circular moves where one particle travels around a simple cycle of the graph. We call the first group of generators the star-type generators and we define them in section \ref{sec:star-gen}. The circular moves are called generators of the loop type and they are defined in section \ref{sec:loop-gen}. Our proof relies on a limited use of discrete Morse theory and combinatorial analysis of certain small canonical graphs. Section \ref{sec:prel} introduces preliminary key definitions from graph theory. In section \ref{sec:morse} we briefly review technical details of the discrete Morse theory. We also tackle the problem of describing relations between the above star and loop generators. This is fully accomplished for star graphs in Proposition \ref{proposition:new relations}. For more complicated graphs, we consider a quotient of the graph braid group where all one-particle moves are put to identity. Physically, this is equivalent to the assumption that the graph is not immersed in any external fields, i.e. there are no Aharonov-Bohm-like effects stemming from situations when one particle travels around a loop in the considered graph. We first analyse the case of $2$-connected graphs in section \ref{sec:2connected} where we show that a sufficiently high connectedness of the graph allows us to greatly reduce the number of generators of the graph braid group. In section \ref{sec:3connected} we consider $3$-connected graphs where our main result is the following theorem. 
\begin{theorem}\label{main theorem}
Let $\graf$ be a $3$-connected planar graph.
Then the braid group $\bfB_n(\graf)$ admits a presentation generated by
\begin{enumerate}
\item $Y$-exchanges,
\item one-particle moves $\{\gamma_0,\hdots, \gamma_{b-1}\}$ for $b=\beta_1(\graf)$ -- the first Betti number of $\Gamma$,
\item one circular move $\delta$.
\end{enumerate}
Moreover, by taking a quotient by all $\gamma_j$'s, we obtain the classical braid group $\bfB_n(\RR^2)$.
\end{theorem}
The above theorem has important conceptual implications for the description of quantum mechanics on graphs. Namely, quantum mechanics on any planar $3$-connected graph can be understood in terms of the $2D$ quantum mechanics. However, as we point out in examples in section \ref{sec:yang-baxter}, this is not the case for $2$-connected graphs where we obtain some generalisations of the $2D$ braiding relation and the Yang-Baxter equation.

Let us finalise this section by a brief review of other approaches to the description of quantum statistics on graphs. One of the earliest works in this area \cite{BE92} tackled the problem of self-adjoint extensions of a multi-particle free hamiltonian defined on the graph configuration space $C_n(\Gamma)$. Importantly, families of self-adjoint extensions turned out to be parametrised by unitary representations of $\bfB_n(\Gamma)$. Similar approach has been utilised in more recent works \cite{Bolte13} to describe self-adjoint extensions of $n$-particle hamiltonians for bosons and fermions with pairwise singular interactions on $\Gamma^n$. Other authors proposed effective discrete hopping hamiltonians for interacting anyons that are defined on discrete configuration spaces of graphs $D_n(\Gamma)\subset C_n(\Gamma)$ \cite{HKR}. For abelian anyons, this led to full classification of abelian quantum statistics on graphs \cite{HKRS} in terms of the first homology group $H_1(C_n(\Gamma),\ZZ)$. For non-abelian quantum statistics, higher homology groups have been used to extract information about classification of flat complex bundles over $C_n(\Gamma)$ \cite{MS17,MS19}. Homology groups of graph configuration spaces as well as presentations of graph braid groups are also subjects of independent interest in mathematics, see \cite{BDKhom,BDKstab,Ramos}.

\section{Preliminaries}\label{sec:prel}

Let $\graf$ be a graph understood as a one-dimensional finite CW-complex. Sets of 0-cells and 1-cells of $\Gamma$ are called vertices and edges respectively. The set of vertices of $\Gamma$  will be denoted by $V$ and  the set of edges by $E$.

We say that two topological graphs $\graf$ and $\graf'$ are \emph{isomorphic} if they are isomorphic as combinatorial graphs, and \emph{homeomorphic} if they are homeomorphic as topological spaces. It is then clear that for each homeomorphic class of a graph $\graf$, there are infinitely many non-isomorphic graphs, which are homeomorphic to $\graf$. However, all of them are related by \emph{subdivision} or \emph{smoothing}.

A \emph{subdivision} $\graf'$ of $\graf$ is a graph obtained by adding one vertex in the middle of a chosen edge of $\Gamma$. In other words, subdivision is the effect of replacing an edge of $\graf$ with linear graph $I_2$ consisting of 3 vertices and two edges. Conversely, we call $\graf$ a \emph{smoothing} of $\graf'$ and we denote it by $\graf'\prec\graf$.

For a given graph $\graf$, we denote the set of isomorphism classes of graphs homeomorphic to $\graf$ by $\cG(\graf)$.
Then $\prec$ defines a partial order in $\cG(\graf)$, which has a unique minimal element $\graf_0$ obtained by smoothing all bivalent vertices of $\graf$.

What is more, some of the graphs in $\cG(\graf)$ can be realised as simplicial complexes. Let $\cG_\Delta(\graf)$ be the subset of $\cG(\graf)$ consisting of simplicial complexes. Then $\cG_\Delta(\graf)$ has a unique minimal element $\graf_\Delta$ with respect to the partial order $\prec$ defined above, which can be obtained from $\graf_0$ by a series of subdivisions according to the following rules.
\begin{enumerate}
\item Subdivide each simple loop twice so that it forms a triangle.
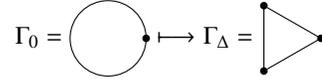
\[
\graf_0=\begin{tikzpicture}[baseline=-.5ex,scale=0.5]
\draw(0,0) circle (1);
\fill(1,0) circle (0.1);
\end{tikzpicture}\longmapsto
\graf_\Delta=
\begin{tikzpicture}[baseline=-.5ex,scale=0.5]
\draw(0:1)--(120:1)--(240:1)--cycle;
\fill(1,0) circle (0.1) (120:1) circle (0.1) (240:1) circle (0.1);
\end{tikzpicture}
\]
\item Subdivide all multiple edges except for one of each.
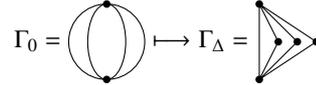
\[
\graf_0=\begin{tikzpicture}[baseline=-.5ex,scale=0.5]
\draw(0,0) circle (1);
\draw(0,0) ellipse (0.5 and 1);
\fill(0,1) circle (0.1) (0,-1) circle (0.1);
\end{tikzpicture}\longmapsto
\graf_\Delta=
\begin{tikzpicture}[baseline=-.5ex,scale=0.5]
\draw(0,1)--(0,-1)--(0.5,0)--(0,1)--(1,0)--(0,-1)--(1.5,0)--(0,1);
\fill(0,1) circle (0.1) (0,-1) circle (0.1) (1.5,0) circle (0.1) (0.5,0) circle (0.1) (1,0) circle (0.1);
\end{tikzpicture}
\]
\end{enumerate}
We call $\graf_\Delta$ the \emph{minimal simplicial representative} of $\cG(\graf)$.

A graph $\graf$ is called \emph{combinatorially $k$-connected} if it has at least two vertices and for any pair $\{v,w\}$ of distinct vertices, there exist $k$ internally-disjoint paths between $v$ and $w$. A singleton graph consisting of one vertex without edges will be regarded as combinatorially $1$-connected but not $k$-connected for any $k\ge 2$. Due to Menger's theorem the combinatorial $k$-connectedness is equivalent to the fact that it is not possible to pick a set of $k-1$ distinct vertices, $K\subset V(\Gamma)$, that disconnects $\Gamma$. In other words, the induced subgraph spanned on vertices $V(\Gamma)\setminus K$ is never disconnected. 

We say that a graph $\graf$ is \emph{topologically $k$-connected} if its minimal simplicial representative $\graf_\Delta$ is combinatorially $k$-connected. In this paper, topological connectivity is the default definition of connectivity, hence topologically $k$-connected graphs will be called just \emph{$k$-connected}.

\begin{example}
The following holds. Let $k\ge 3$.
\begin{itemize}
\item Star graphs $S_k$ are $1$-connected but not $2$-connected.
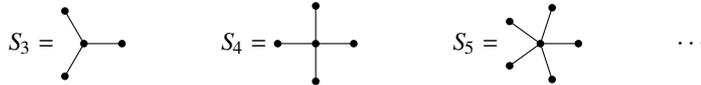
\begin{align*}
S_3&=\begin{tikzpicture}[baseline=-.5ex,scale=0.5]
\fill (0,0) circle (0.1);
\foreach \i in {0,1,2} {
\draw(0,0) --({120*\i}:1);
\fill({120*\i}:1) circle (0.1);
}
\end{tikzpicture}&
S_4&=\begin{tikzpicture}[baseline=-.5ex,scale=0.5]
\fill (0,0) circle (0.1);
\foreach \i in {0,...,3} {
\draw(0,0) --({90*\i}:1);
\fill({90*\i}:1) circle (0.1);
}
\end{tikzpicture}&
S_5&=\begin{tikzpicture}[baseline=-.5ex,scale=0.5]
\fill (0,0) circle (0.1);
\foreach \i in {0,...,4} {
\draw(0,0) --({72*\i}:1);
\fill({72*\i}:1) circle (0.1);
}
\end{tikzpicture}&
\cdots
\end{align*}
\item Cyclic graphs $C_k$ are $2$-connected but not $3$-connected.
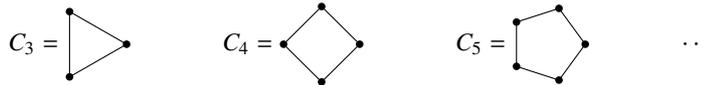
\begin{align*}
C_3&=\begin{tikzpicture}[baseline=-.5ex,scale=0.5]
\foreach \i in {0,1,2} {
\draw({120*(\i+1)}:1) --({120*\i}:1);
\fill({120*\i}:1) circle (0.1);
}
\end{tikzpicture}&
C_4&=\begin{tikzpicture}[baseline=-.5ex,scale=0.5]
\foreach \i in {0,...,3} {
\draw({90*(\i+1)}:1) --({90*\i}:1);
\fill({90*\i}:1) circle (0.1);
}
\end{tikzpicture}&
C_5&=\begin{tikzpicture}[baseline=-.5ex,scale=0.5]
\foreach \i in {0,...,4} {
\draw({72*(\i+1)}:1) --({72*\i}:1);
\fill({72*\i}:1) circle (0.1);
}
\end{tikzpicture}&
\cdots
\end{align*}
\item Graphs $\Theta_k$ are topologically $2$-connected but not $3$-connected.
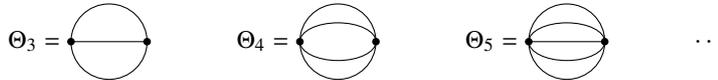
\begin{align*}
\Theta_3&=\begin{tikzpicture}[baseline=-.5ex,scale=0.5]
\fill(-1,0) circle (0.1) (1,0) circle (0.1);
\draw(-1,0)--(1,0) (0,0) circle (1);
\end{tikzpicture}&
\Theta_4&=\begin{tikzpicture}[baseline=-.5ex,scale=0.5]
\fill(-1,0) circle (0.1) (1,0) circle (0.1);
\draw (0,0) ellipse (1 and 0.5);
\draw (0,0) circle (1);
\end{tikzpicture}&
\Theta_5&=\begin{tikzpicture}[baseline=-.5ex,scale=0.5]
\fill(-1,0) circle (0.1) (1,0) circle (0.1);
\draw(-1,0)--(1,0) (0,0) circle (1);
\draw (0,0) ellipse (1 and 0.5);
\end{tikzpicture}&
\cdots
\end{align*}
\item Wheel graphs $W_k$ are $3$-connected.
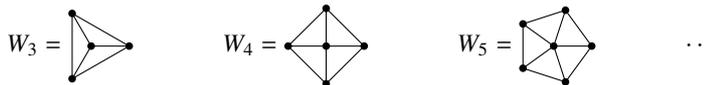
\begin{align*}
W_3&=\begin{tikzpicture}[baseline=-.5ex,scale=0.5]
\fill (0,0) circle (0.1);
\foreach \i in {0,1,2} {
\draw({120*(\i+1)}:1) --({120*\i}:1);
\fill({120*\i}:1) circle (0.1);
\draw(0,0) --({120*\i}:1);
}
\end{tikzpicture}&
W_4&=\begin{tikzpicture}[baseline=-.5ex,scale=0.5]
\fill (0,0) circle (0.1);
\foreach \i in {0,...,3} {
\draw({90*(\i+1)}:1) --({90*\i}:1);
\fill({90*\i}:1) circle (0.1);
\draw(0,0) --({90*\i}:1);
}
\end{tikzpicture}&
W_5&=\begin{tikzpicture}[baseline=-.5ex,scale=0.5]
\fill (0,0) circle (0.1);
\foreach \i in {0,...,4} {
\draw({72*(\i+1)}:1) --({72*\i}:1);
\fill({72*\i}:1) circle (0.1);
\draw(0,0) --({72*\i}:1);
}
\end{tikzpicture}&
\cdots
\end{align*}
\item Any topologically non-$2$-connected graph has a vertex $v_0$ such that $\Gamma\setminus\{v_0\}$ is not connected.
\end{itemize}
\end{example}

Let $\graf$ be a graph. In further parts of this paper, we will often consider the space $\tilde \graf^{S_3}\coloneqq \{\iota:S_3\to \graf\}$ of all topological embeddings from $S_3$ to $\graf$.
\[
\left(\iota:S_3=\begin{tikzpicture}[baseline=-.5ex,scale=0.5]
\fill (0,0) circle (0.1) (0:1) circle (0.1) (120:1) circle (0.1) (240:1) circle (0.1);
\draw(0,0) -- (120:1) (0,0)--(240:1) (0,0)--(1,0);
\end{tikzpicture}\to\graf\right)\in \tilde\graf^{S_3}.
\]
This means that embedding $\iota$ may not the map the endpoint each leaf of $S_3$ to a vertex of $\graf$. However, it must send the central vertex of $S_3$ to an essential vertex of $\graf$. We say that two such embeddings $\iota$ and $\iota'$ are \emph{equivalent} if they are isotopic as embeddings (up to automorphisms of $S_3$) between topological spaces.
We denote the resulting quotient space by $\graf^{S_3}$. As a result, $\graf^{S_3}$ is a finite discrete set whose each element can be (uniquely) characterised by a pair $(v,\bfe=\{e_1,e_2,e_3\})$, where $v$ is a vertex and $e_i$'s are three distinct edges adjacent to $v$.

\subsection{Presentations of the braid group for particles in $\RR^2$}

In this paper we will often make references to different presentations of the planar braid group $\bfB_n(\RR^2)$. Let us next briefly review the relevant presentations. As a general reference for this subsection, we direct the reader to \cite{braids}. We start with the Artin presentation of $\bfB_n(\RR^2)$. In this presentation, generators of $\bfB_n(\RR^2)$ are called simple braids. We assume that the initial configuration is such that particles $1,\hdots, n$ are aligned on a line segment in $\RR^2$ in an equidistant way. Simple braid $\sigma_i$ exchanges particle $i$ with particle $i+1$ in a clockwise direction while other particles remain in their fixed positions (the choice of anti-clockwise exchange is a matter of convention). In this way, we obtain the set of $n-1$ generators $\sigma_1,\hdots,\sigma_{n-1}$. Clearly, simple braids involving disjoint sets of particles commute with each other, i.e. $\sigma_i\sigma_j=\sigma_j\sigma_i$ for $|i-j|\geq 2$. Importantly, simple braids involving triples of consecutive particles satisfy the braiding relation, i.e. $\sigma_i\sigma_{i+1}\sigma_i=\sigma_{i+1}\sigma_i\sigma_{i+1}$ for $i\in \{1,\hdots, n-2\}$. 

When applied to a special class of unitary representations of $\bfB_n(\RR^2)$, the braiding relation leads to the celebrated Yang-Baxter equation. This concerns high-dimensional representations of $\bfB_n(\RR^2)$ to unitary operators acting on the Hilbert of $n$ particles. We assume that the Hilbert space of a single particle is finite-dimensional, $\mH_1=\CC^d$. Recall that the $n$-particle Hilbert space is the $n$-fold tensor product of single-particle spaces, i.e. $\mH_n=\left(\CC^d\right)^{\otimes n}$. The Yang-Baxter equation is an equation for a unitary operator $R:\ \CC^d\otimes \CC^d\to \CC^d\otimes \CC^d$, called the $R$-matrix, that satisfies
\begin{equation}\label{yang-baxter}
(R\otimes\bone)(\bone\otimes R)(R\otimes\bone)=(\bone\otimes R)(R\otimes\bone)(\bone\otimes R).
\end{equation}
The corresponding representation $\bfB_n(\RR^2)\to U(d^n)$ is constructed by assigning
\[\sigma_i\mapsto U_i:=\bone\otimes\hdots\otimes\bone\otimes R_{i,i+1}\otimes\bone\otimes\hdots\otimes\bone.\]
Because the $R$-matrix satisfies equation (\ref{yang-baxter}), operators $U_i$ automatically satisfy the braiding relation. We will revisit this construction in section \ref{sec:yang-baxter}.

Another presentation of $\bfB_n(\RR^2)$ that will be crucial for this paper uses only two generators. This is done by introducing the so-called total braid $\delta:=\sigma_1\sigma_2\hdots\sigma_{n-1}$. By utilising the commutative and braiding relations for simple braids, one can verify that $\sigma_{i+1}=\delta\sigma_i\delta^{-1}$. This allows one to recursively express any simple braid as $\sigma_{i+1}=\delta^{i}\sigma_1\delta^{-i}$. When expressed in terms of generators $\delta$ and $\sigma_1$, braiding and commutative relations have the following forms
\begin{gather*}
\sigma_1\delta\sigma_1\delta^{-1}\sigma_1=\delta\sigma_1\delta^{-1}\sigma_1\delta\sigma_1\delta^{-1}, \\
\sigma_1\delta^{j-i}\sigma_1\delta^{i-j}=\delta^{j-i}\sigma_1\delta^{i-j}\sigma_1{\quad\mathrm{\ for\ }}|i-j|\geq2.
\end{gather*}

\section{Presentations of graph braid groups}
Throughout this section, we assume that i) graph $\graf$ is finite, connected and has at least one essential vertex, ii)  $\graf$ is {\textit{sufficiently subdivided}} with respect to a fixed $n>1$. This means that the edge-length of every path between two essential vertices is at least $n-1$ and the edge-length of every loop is at least $n+1$.

\subsection{Morse presentations}\label{sec:morse}
We first review Farley-Sabalka's algorithm \cite{FSbraid} which provides a group presentation of the graph braid group for any graph  in terms of cells of an appropriate Morse complex. Certain features of Morse presentations of graph braid groups will be key ingredients of some of our proofs in later subsections. Farley-Sabalka's algorithm realises a specific instance of Forman's discrete Morse theory \cite{Forman} on the {\textit{discrete configuration space}}, $D_n(\Gamma)$. A discrete configuration space is a regular cube complex that is a deformation retract of $C_n(\Gamma)$. For the deformation-retracion $C_n(\Gamma)\to D_n(\Gamma)$ to work, $\graf$ has to be sufficiently subdivided with respect to the particle number $n$ \cite{Abrams,Abrams-subdiv}. Cells of $D_n(\Gamma)$ have the following form
\[
c=\{c_1,c_2,\hdots,c_n\}\subset \graf\mid \bar c_i\cap \bar c_j=\emptyset\quad \forall i\neq j
\]
where $\{c_i\}_{i=1}^n$ are mutually disjoint closed cells of $\graf$, i.e. they are either edges or vertices. Cell $c$ can be viewed as a cube
\[
c\cong \bigtimes_{i=1}^n c_i.
\]
Hence the dimension of $c$ is given by $|E(\Gamma)\cap c|$. In this way, the discrete configuration space can also be viewed as a proper subset of $C_n(\Gamma)$. 

The next step is the construction of a \emph{rooted spanning tree} that determines the ordering of vertices of $\graf$.
We choose a rooted spanning tree $(T,*)\subset \graf$ such that $*$ is univalent in $T$. Edges that belong to the complement of $T$ in $\Gamma$ are called \emph{deleted edges}. We assume that $T$ is chosen so that the endpoints of deleted edges are always bivalent vertices.

We fix a planar embedding $T\subset \RR^2$ and label vertices in $T$ as follows.
\begin{itemize}
\item The root $*$ is labelled by $0$.
\item To determine labels of the remaining vertices, we imagine a thin ribbon neighbourhood of $T$. Starting from $*$ and travelling clockwise along the boundary of the ribbon, we label the visited vertices as $1,2,\hdots, |V|-1$.
\end{itemize}
We denote the deleted edge that is adjacent to $*$ by $e_0$ (if such an edge exists). Between any two vertices $\{v_1,v_2\}\subset V(\Gamma)$ there exists a unique edge path from $v_1,v_2$ in $T$. Such a path will be denoted by $[v_1,v_2]$. Therefore, by choosing a rooted spanning tree, we have a linear order on $V(\graf)$ that we denote by $<$. For any edge $e\in E(\Gamma)$ we define its initial and terminal vertices denoted by $\iota(e)$ and $\tau(e)$ respectively so that $\iota(e)>\tau(e)$. Similarly, for any $v\in V(\graf)$ we define $e(v)$ as the edge for which $\iota(e(v))=v$. For an embedding between labelled graphs $f:\graf_1\to\graf_2$, we say that $f$ is \emph{order-preserving} if it preserves the order of vertices.

\begin{example}
For a subdivided $\Theta_k$, we may choose a rooted spanning tree $(T_k,*)$ as follows:
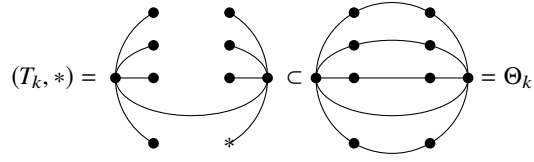
\[
(T_k,*)=\begin{tikzpicture}[baseline=-.5ex]
\fill(-1,0) circle (2pt) (1,0) circle (2pt);
\draw(-1,0)--(-0.5,0) (0.5,0)--(1,0) (60:1) arc (60:-60:1) node {$*$} (120:1) arc (120:240:1);
\draw (60:1 and 0.5) arc (60:-240:1 and 0.5);
\fill(60:1 and 0.5) circle (2pt) (60:1) circle (2pt) (120:1) circle (2pt) (240:1) circle (2pt) (120:1 and 0.5) circle (2pt) (0.5,0) circle (2pt) (-0.5,0) circle (2pt);
\end{tikzpicture}
\subset\begin{tikzpicture}[baseline=-.5ex]
\fill(-1,0) circle (2pt) (1,0) circle (2pt);
\fill(60:1 and 0.5) circle (2pt) (60:1) circle (2pt) (120:1) circle (2pt) (240:1) circle (2pt) (120:1 and 0.5) circle (2pt) (0.5,0) circle (2pt) (-0.5,0) circle (2pt) (-60:1) circle (2pt);
\draw(-1,0)--(1,0) (0,0) circle (1);
\draw (0,0) ellipse (1 and 0.5);
\end{tikzpicture}=\Theta_k
\]
Then the ribbon of $T_k$ and labels are given below.
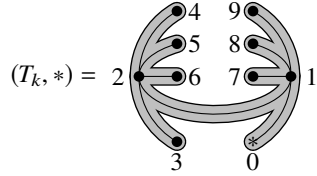
\[
(T_k,*)=\begin{tikzpicture}[baseline=-.5ex]
\draw[line width=7,cap=round](-1,0)--(-0.5,0) (0.5,0)--(1,0) (60:1) arc (60:-60:1) (120:1) arc (120:240:1);
\draw[line width=7,cap=round] (60:1 and 0.5) arc (60:-240:1 and 0.5);
\draw[line width=6,cap=round,lightgray](-1,0)--(-0.5,0) (0.5,0)--(1,0) (60:1) arc (60:-60:1) (120:1) arc (120:240:1);
\draw[line width=6,cap=round,lightgray] (60:1 and 0.5) arc (60:-240:1 and 0.5);
\draw(-1,0)--(-0.5,0) (0.5,0)--(1,0) (60:1) arc (60:-60:1) node {$*$} node[below=1pt] {$0$} (120:1) arc (120:240:1);
\draw(60:1 and 0.5) arc (60:-240:1 and 0.5);
\fill(-1,0) circle (2pt) node[left=2pt] {$2$} (1,0) circle (2pt) node[right=2pt] {$1$};
\fill(60:1 and 0.5) circle (2pt) node[left=1pt] {$8$} (60:1) circle (2pt) node[left=1pt] {$9$} (120:1) circle (2pt) node[right=1pt] {$4$} (240:1) circle (2pt) node[below=1pt] {$3$} (120:1 and 0.5) circle (2pt) node[right=1pt] {$5$} (0.5,0) circle (2pt) node[left=1pt] {$7$} (-0.5,0) circle (2pt) node[right=1pt] {$6$};
\end{tikzpicture}
\]
\end{example}

In the description of the graph braid group, we are particularly interested in loops and paths in $C_n(\Gamma)$. Note first that any path in $C_n(\Gamma)$ with its endpoints contained in $D_n(\Gamma)\subset C_n(\Gamma)$ is homotopy equivalent to a path that is entirely contained in $D_n(\Gamma)$. Furthermore, any path in $D_n(\Gamma)$ is homotopy equivalent to a path that is fully contained in the $1$-skeleton of $D_n(\Gamma)$. In particular, any element $\bfB_n(\Gamma)$ can be represented as a word of $1$-cells on $D_n(\Gamma)$. More precisely, cell $\{e,v_1,\hdots,v_{n-1}\}$ is viewed as a directed path from $\{\iota(e),v_1,\hdots,v_{n-1}\}$ to $\{\tau(e),v_1,\hdots,v_{n-1}\}$. 

The general idea of the discrete Morse theory is to contract configuration space $D_n(\Gamma)$ to a much smaller Morse complex $M_n(\Gamma)$ whose $1$-skeleton is a wedge of circles based at the same point given by configuration $\{0,1,\hdots,n-1\}$. Then, one uses the fact that $\bfB_n(\Gamma)\cong\pi_1\left(M_n(\Gamma)\right)$, i.e. $\bfB_n(\Gamma)$ is generated by the circles in the $1$-skeleton of $M_n(\Gamma)$.

The construction of discrete Morse theory relies on the notion of {\it Morse matching}. A Morse matching is a collection of partially defined maps $\{W_d\}_{d=0}^{n}$ where $W_i$ maps some of $d$-cells in $D_n(\Gamma)$ to some of $(d+1)$-cells. In Farley-Sabalka's algorithm the Morse matching can be understood in terms of imposing certain rules for particle movement on $\graf$. A $d$-dimensional cell $c\in D_n(\Gamma)$ is viewed as a move where $d$ particles simultaneously slide along edges $e\in E(\Gamma)\cap c$ in the direction from $\iota(e)$ to $\tau(e)$. We say that an edge $e\in c$ is $\emph{order-respecting}$ if $e\subset T$ and there are no vertices $v\in c$ for which $v<\tau(e)$ and $\tau(e(v))=\tau(e)$. Intuitively, on junctions in $\graf$ only particles that have no other particles to their right (in terms of vertex ordering on $T$) are allowed to move first. In a similar spirit one defines \emph{blocked vertices} as vertices whose movement is blocked in $c$ by some other particles. Vertex $v\in c$ is blocked if there exists $c_i\in c$ such that $\tau(e(v))\cap c_i\neq\emptyset$. Clearly, any edge $e\in c$ is either order-respecting or non order-respecting. Similarly, $v\in c$ can be either blocked or unblocked. 
\begin{definition}
A cell $c\in D_n(\Gamma)$ is critical if and only all vertices $v\in V(\Gamma)\cap c$ are blocked and all edges $e\in E(\Gamma)\cap c$ are non order-respecting.
\end{definition}
The definition of map $W_d$ is recursive. Namely, the domain of $W_d$ consists of cells that do not belong to the image of $W_{d-1}$ and are not critical. For such cells, $W_d$ finds the lowest unblocked vertex in $c$, $v_0$, and replaces it with edge $e(v)$, i.e.
\[W_d(c)=\left(c\setminus\{v_0\}\right)\cup\{e(v_0)\}.\]
Cells that belong to the domain of some $W_d$ are called \emph{redundant} and cells that belong to the image of some $W_{d}$ are called \emph{collapsible}. The Morse complex $M_n(\Gamma)$ is a complex whose $d$-skeleton consists of critical $d$-cells in $D_n(\Gamma)$. Importantly, any $d$-cell $c$ in can be mapped to $M_n(\Gamma)$ as a word that consists only of critical cells. This is done by means of the \emph{principal reduction} $F$. We will specify this construction only for $1$-cells as this is the only relevant case for this paper. To this end, consider the boundary word of $2$-cell $W_1(c)$ for a redundant cell $c$ and bring it to the form $cw$ where $w$ is a word of appropriate $1$-cells. Then, we define the principal reduction of a redundant $c$ as $F(c)=w^{-1}$. The action of $F$ is extended to critical sells as $F(c)=c$ and to collapsible cells as $F(c)=1$. We extend $F$ to any word of $1$-cells in a natural way. By applying map $F$ to any word sufficiently many times, we always end up with a word that consists only of critical cells. Such a word is invariant under map $F$. Let us denote such an infinite iteration of $F$ by $F^\infty$. We are now ready to state the central theorem of this subsection.
\begin{theorem}[Farley-Sabalka]\label{thm:FS}
Group $\bfB_n(\graf)$ is generated by critical $1$-cells in $D_n(\Gamma)$. Relators are given by $F^\infty(w)$ where $w$ goes through boundary words of all critical $2$-cells in $D_n(\Gamma)$.
\end{theorem}

Recall that any critical $1$-cell has the form $c=\{e,v_1,\hdots,v_{n-1}\}$ where $e$ is a non order-respecting edge and all $\{v_i\}_{i=1}^{n-1}$ are blocked. We will distinguish the following two types of critical $1$-cells.
\begin{enumerate}
\item Critical cells associated with an essential vertex, i.e. $e\in T$. 
\item Critical cells associated with a deleted edge, i.e. $e\in \graf\setminus T$. 
\end{enumerate}
A critical cell at an essential vertex has the following properties i) $\tau(e)$ is an essential vertex and there exists $i$ such that $\tau(e(v_i))=\tau(e)$ and $\tau(e)<v_i<\iota(e)$, ii) for every $i$, we have $\tau(e(v_i))=v_j$ for some $i\neq j$ or $\tau(e(v_i))=\iota(e)$. Hence, the form of such a critical $1$-cell can be uniquely determined by specifying i) the relevant essential vertex, ii) the distribution of particles on leaves incident at the essential vertex and iii) the leaf at $v$ containing the non-order respecting edge in $c$. If the essential vertex is $v$ and $v$ is of valence $k$, such a critical cell will be denoted by $v_i(\bfb)$, where $i\in\{1,\hdots,k-1\}$ is the leaf containing the non order-respecting edge and $\bfb=(b_1,b_2,\hdots,b_{k-1})$ is a sequence of nonnegative integers that specifies the distribution of particles on leaves at vertex $v$.

Furthermore, the assumption that end-vertices of all deleted edges are bivalent implies that all critical cells that are associated with a deleted edge $e$ are of the form $\{e,0,1,\hdots,n-2\}$ if $\tau(e)\neq*$ and $\{e,1,2,\hdots,n-1\}$ otherwise.

Note also that under the above assumptions about the choice of the spanning tree, there is just one critical $0$-cell in $D_n(\Gamma)$, namely $\{0,1,\hdots,n-1\}$. Hence, the $1$-skeleton of $M_n(\Gamma)$ consists of $1$-cells that are topological circles based at the unique $0$-cell.

 In the following subsections we establish a correspondence between geometric generators of $\bfB_n(\graf)$ and the above two types of critical $1$-cells. Namely, critical cells at essential vertices will correspond to generators of the \emph{star type} and critical cells associated with deleted edges will correspond to generators of the \emph{loop type}.

\subsection{Generators of the star type}\label{sec:star-gen}
One of the building blocks in geometric presentations of graph braid groups is an exhaustive description of $\bfB_n(\graf)$ when $\graf=S_k$, a star graph with $k$ leaves. This is because for any graph $\graf$ generators of $\bfB_n(S_k)$ can be regarded as generators of $\bfB_n(\graf)$. To see this, for an essential vertex $v\in V(\graf)$ of valency $k$ consider the following order-preserving embedding of $S_k$ into spanning tree $T\subset \graf$ centred at $v$
\[
\iota_v:(S_k,*)=\begin{tikzpicture}[baseline=-.5ex,scale=0.5]
\fill (0,0) circle (0.1);
\foreach \i in {0,...,4} {
\draw(0,0) --({72*\i}:1);
\fill({72*\i}:1) circle (0.1);
}
\draw (-72:1) node[right] {$*$};
\end{tikzpicture}\to (T,*)\subset\graf,
\]
where the central vertex of $S_k$ is mapped to $v$, the $0$-th edge of $S_k$ is mapped to a path from $*$ to $v$ and the $i$-th edge of $S_k$ ($0<i<k$) is mapped to the corresponding (subdivided) edge adjacent to $v$. For such an embedding, critical $1$-cells of $D_n(\graf)$ at $v$ are in a one-to-one correspondence with critical $1$-cells of $D_n(S_k)$. Hence, $\iota_v$ indices a well-defined map from $\bfB_n(S_k)$ to $\bfB_n(\graf)$.

In the remaining part of this subsection we will introduce geometric generators of the star type and by writing down all relations between them we will reproduce the following well-known result.
\begin{theorem}\cite{Ghirst,KP}
The $n$-braid group $\bfB_n(S_k)$ of the star graph $S_k$ is a free group of rank 
\[
N_1(n,k)\coloneqq(k-2)\binom{n+k-2}{n-1}-\binom{n+k-2}{n}+1.
\]
\end{theorem}

The strategy for this subsection is to find for each critical $1$-cell $c\in D_n(S_k)$ its corresponding loop $\gamma\subset C_n(\Gamma)$ such that the associated word of $1$-cells, $w_\gamma$, satisfies $F^\infty(w_\gamma)=c$. To this end, we introduce a shorthand notation for the configuration where $a_i$ points occupies $i$th leaf of $S_k$. We denote the leaves by $e_0,\hdots,e_{k-1}$ and the respective particle configuration by by $e_0^{a_0}e_1^{a_1}\cdots e_k^{a_k}$.
\begin{align*}
e_0^7&=\begin{tikzpicture}[baseline=-.5ex,scale=0.5,xscale=-1]
\fill[gray](0,0) circle (2pt);
\foreach \i in {1,...,4} {
\fill[gray]({\i*72}:2) circle (2pt);
\draw(0,0) -- ({\i*72}:2);
\draw({\i*72}:2.4) node {$\i$};
}
\draw(0,0) -- (0:4) node[below] {$0$};
\fill(0:4) circle (2pt);
\foreach \i in {2,...,8} {
\fill[red]({\i*0.5},0) circle (0.1);
}
\end{tikzpicture}&
e_0^4e_1^2e_3&=\begin{tikzpicture}[baseline=-.5ex,scale=0.5,xscale=-1]
\fill[gray](0,0) circle (2pt);
\foreach \i in {1,...,4} {
\fill[gray]({\i*72}:2) circle (2pt);
\draw(0,0) -- ({\i*72}:2);
\draw({\i*72}:2.4) node {$\i$};
}
\draw(0,0) -- (0:2) node[below] {$0$};
\fill[gray](0:2) circle (2pt);
\foreach \i in {1,...,4} {
\fill[red]({\i*0.5},0) circle (0.1);
}
\fill[red](72:2) circle (0.1) (216:2) circle (0.1) (72:1.5) circle (0.1);
\end{tikzpicture}
\end{align*}

We also denote by $\beta^i$ the path from configuration $\bfx e_0$ to configuration $\bfx e_i$ for some abstract configuration $\bfx$ where the top particle fom leaf $e_0$ is moved to leaf $i$. Consequently, by $\beta^{-i}$ we denote the inverse of $\beta^i$.
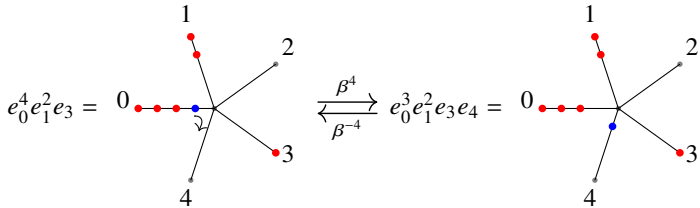
\[
\begin{tikzcd}
e_0^4e_1^2e_3=\begin{tikzpicture}[baseline=-.5ex,scale=0.5,xscale=-1]
\fill[gray](0,0) circle (2pt);
\foreach \i in {0,...,4} {
\fill[gray]({\i*72}:2) circle (2pt);
\draw(0,0) -- ({\i*72}:2);
\ifthenelse{\i = 4}
{\draw({\i*72}:2.4) node [below=-6pt] {$\i$}}
{\draw({\i*72}:2.4) node {$\i$}};
}
\fill[red](2,0) circle (0.1) (1.5,0) circle (0.1) (72:2) circle (0.1) (72:1.5) circle (0.1) (216:2) circle (0.1) (1,0) circle (0.1);
\fill[blue] (0.5,0) circle (0.1);
\draw[rounded corners,<-] (-72:0.7)++(18:0.2) -- (-36:0.3) -- (0.5,-0.2);
\end{tikzpicture}\arrow[r,"\beta^4",yshift=.5ex]&
e_0^3e_1^2e_3e_4=\begin{tikzpicture}[baseline=-.5ex,scale=0.5,xscale=-1]
\fill[gray](0,0) circle (2pt);
\foreach \i in {0,...,4} {
\fill[gray]({\i*72}:2) circle (2pt);
\draw(0,0) -- ({\i*72}:2);
\ifthenelse{\i = 4}
{\draw({\i*72}:2.4) node [below=-6pt] {$\i$}}
{\draw({\i*72}:2.4) node {$\i$}};
}
\fill[red](2,0) circle (0.1) (1.5,0) circle (0.1) (72:2) circle (0.1) (72:1.5) circle (0.1) (216:2) circle (0.1) (1,0) circle (0.1);
\fill[blue] (-72:0.5) circle (0.1);
\end{tikzpicture}\arrow[l,"\beta^{-4}",yshift=-.5ex]
\end{tikzcd}
\]

\begin{definition}[$Y$-exchange]
For a sequence $\bfa=(a_1,\hdots, a_{i-1})$ with $1\le a_j<k$, we define $\beta^\bfa$ as the concatenation
\[
\beta^\bfa = \beta^{a_1}\cdots\beta^{a_{i-1}}. 
\]
For $1\le a,b <k$, we define a loop, called a \emph{$Y$-exchange} on leaves $a,b$ as
\[
\sigma_i^{\bfa,a,b}:=\beta^{\bfa}Y^{ab}\beta^{-\bfa}:e_0^n\to e_0^n\in \bfB_n(S_k),
\]
where 
\[
-\bfa = (-a_{i-1},-a_{i-2},\hdots, -a_1),
\]
and $Y^{ab}$ is the commutator of $\beta^a$ and $\beta^b$
\[
Y^{ab} := \beta^{a,b,-a,-b}=\beta^a\beta^b\beta^{-a}\beta^{-b}.
\]
In particular,
\[
\sigma_1^{a,b}=Y^{ab}.
\]
\end{definition}
\begin{remark}
We will also use the notation $\sigma_i^\bfa$ when $\bfa=(a_1,\hdots, a_{i+1})$ to denote 
$\sigma_i^{\bfa',a_i,a_{i+1}}$ for $\bfa'=(a_1,\hdots, a_{i-1})$.
\end{remark}

\begin{example}
Paths $\beta^{1,3,1}$ and $Y^{4,2}$ for respective starting configurations $e_0^7$ and $e_0^4e_1^2e_3^1$ are depicted in Figure~\ref{figure:paths}.
\end{example}

\begin{figure}[ht]
\[
\begin{tikzcd}
e_0^7=\begin{tikzpicture}[baseline=-.5ex,scale=0.5,xscale=-1]
\fill[gray](0,0) circle (2pt);
\foreach \i in {1,...,4} {
\fill[gray]({\i*72}:2) circle (2pt);
\draw(0,0) -- ({\i*72}:2);
\ifthenelse{\i = 4}
{\draw({\i*72}:2.4) node [below=-6pt] {$\i$}}
{\draw({\i*72}:2.4) node {$\i$}};
}
\draw(0,0) -- (0:4) node[below] {$0$};
\fill(0:4) circle (2pt);
\foreach \i in {2,...,8} {
\fill[red]({\i*0.5},0) circle (0.1);
}
\draw[rounded corners,<-] (72:2)++(-18:0.2) -- (36:0.3) -- (1,0.2);
\end{tikzpicture}\arrow[r,"\beta^1"] \arrow[d,"\beta^{1,3,1}"']&
\begin{tikzpicture}[baseline=-.5ex,scale=0.5,xscale=-1]
\fill[gray](0,0) circle (2pt);
\foreach \i in {1,...,4} {
\fill[gray]({\i*72}:2) circle (2pt);
\draw(0,0) -- ({\i*72}:2);
\ifthenelse{\i = 4}
{\draw({\i*72}:2.4) node [below=-6pt] {$\i$}}
{\draw({\i*72}:2.4) node {$\i$}};
}
\draw(0,0) -- (0:4) node[below] {$0$};
\fill(0:4) circle (0.1);
\foreach \i in {3,...,8} {
\fill[red]({\i*0.5},0) circle (0.1);
}
\fill[red](72:2) circle (0.1);
\draw[rounded corners,<-] (216:2)++(306:0.2) -- (-72:0.3) -- (1.5,-0.2);
\end{tikzpicture}\arrow[d,"\beta^3"]\\
e_0^4e_1^2e_3=\begin{tikzpicture}[baseline=-.5ex,scale=0.5,xscale=-1]
\fill[gray](0,0) circle (2pt);
\foreach \i in {1,...,4} {
\fill[gray]({\i*72}:2) circle (2pt);
\draw(0,0) -- ({\i*72}:2);
\ifthenelse{\i = 4}
{\draw({\i*72}:2.4) node [below=-6pt] {$\i$}}
{\draw({\i*72}:2.4) node {$\i$}};
}
\draw(0,0) -- (0:2) node[below] {$0$};
\fill[gray](0:2) circle (2pt);
\foreach \i in {1,...,4} {
\fill[red]({\i*0.5},0) circle (0.1);
}
\fill[red](72:2) circle (0.1) (216:2) circle (0.1) (72:1.5) circle (0.1);
\end{tikzpicture}
&
\begin{tikzpicture}[baseline=-.5ex,scale=0.5,xscale=-1]
\fill[gray](0,0) circle (2pt);
\foreach \i in {1,...,4} {
\fill[gray]({\i*72}:2) circle (2pt);
\draw(0,0) -- ({\i*72}:2);
\ifthenelse{\i = 4}
{\draw({\i*72}:2.4) node [below=-6pt] {$\i$}}
{\draw({\i*72}:2.4) node {$\i$}};
}
\draw(0,0) -- (0:4) node[below] {$0$};
\fill[gray](0:4) circle (2pt);
\foreach \i in {4,...,8} {
\fill[red]({\i*0.5},0) circle (0.1);
}
\fill[red](72:2) circle (0.1) (216:2) circle (0.1);
\draw[rounded corners,<-] (72:1.5)++(-18:0.2) -- (36:0.3) -- (2,0.2);
\end{tikzpicture}\arrow[l,"\beta^1"']
\end{tikzcd}
\]
\[
Y^{4,2}=\left(
\begin{tikzcd}
\bfx=\begin{tikzpicture}[baseline=-.5ex,scale=0.5,xscale=-1]
\fill[gray](0,0) circle (2pt);
\foreach \i in {0,...,4} {
\fill[gray]({\i*72}:2) circle (2pt);
\draw(0,0) -- ({\i*72}:2);
\ifthenelse{\i = 4}
{\draw({\i*72}:2.4) node [below=-6pt] {$\i$}}
{\draw({\i*72}:2.4) node {$\i$}};
}
\fill[red](2,0) circle (0.1) (1.5,0) circle (0.1) (72:2) circle (0.1) (72:1.5) circle (0.1) (216:2) circle (0.1);
\fill[blue] (0.5,0) circle (0.1) (1,0) circle (0.1);
\draw[rounded corners,<-] (-72:0.7)++(18:0.2) -- (-36:0.3) -- (0.5,-0.2);
\end{tikzpicture}\arrow[r,"\beta^4"]&
\begin{tikzpicture}[baseline=-.5ex,scale=0.5,xscale=-1]
\fill[gray](0,0) circle (2pt);
\foreach \i in {0,...,4} {
\fill[gray]({\i*72}:2) circle (2pt);
\draw(0,0) -- ({\i*72}:2);
\ifthenelse{\i = 4}
{\draw({\i*72}:2.4) node [below=-6pt] {$\i$}}
{\draw({\i*72}:2.4) node {$\i$}};
}
\fill[red](2,0) circle (0.1) (1.5,0) circle (0.1) (72:2) circle (0.1) (72:1.5) circle (0.1) (216:2) circle (0.1);
\fill[blue] (-72:0.5) circle (0.1) (1,0) circle (0.1);
\draw[rounded corners,<-] (144:0.5)++(54:0.2) -- (72:0.3) -- (1,0.2);
\end{tikzpicture}\arrow[d,"\beta^2"]\\
\begin{tikzpicture}[baseline=-.5ex,scale=0.5,xscale=-1]
\fill[gray](0,0) circle (2pt);
\foreach \i in {0,...,4} {
\fill[gray]({\i*72}:2) circle (2pt);
\draw(0,0) -- ({\i*72}:2);
\ifthenelse{\i = 4}
{\draw({\i*72}:2.4) node [below=-6pt] {$\i$}}
{\draw({\i*72}:2.4) node {$\i$}};
}
\fill[red](2,0) circle (0.1) (1.5,0) circle (0.1) (72:2) circle (0.1) (72:1.5) circle (0.1) (216:2) circle (0.1);
\fill[blue] (144:0.5) circle (0.1) (1,0) circle (0.1);
\draw[rounded corners,->] (144:0.5)++(54:0.2) -- (72:0.3) -- (0.5,0.2);
\end{tikzpicture}\arrow[u,"\beta^{-2}"] &
\begin{tikzpicture}[baseline=-.5ex,scale=0.5,xscale=-1]
\fill[gray](0,0) circle (2pt);
\foreach \i in {0,...,4} {
\fill[gray]({\i*72}:2) circle (2pt);
\draw(0,0) -- ({\i*72}:2);
\ifthenelse{\i = 4}
{\draw({\i*72}:2.4) node [below=-6pt] {$\i$}}
{\draw({\i*72}:2.4) node {$\i$}};
}
\fill[red](2,0) circle (0.1) (1.5,0) circle (0.1) (72:2) circle (0.1) (72:1.5) circle (0.1) (216:2) circle (0.1);
\fill[blue] (-72:0.5) circle (0.1) (144:0.5) circle (0.1);
\draw[rounded corners,->] (-72:0.5)++(18:0.2) -- (-36:0.3) -- (1,-0.2);
\end{tikzpicture}\arrow[l,"\beta^{-4}"']
\end{tikzcd}
\right)
\]
\caption{Paths $\beta^{1,3,1}$ and $Y^{4,2}$.}
\label{figure:paths}
\end{figure}
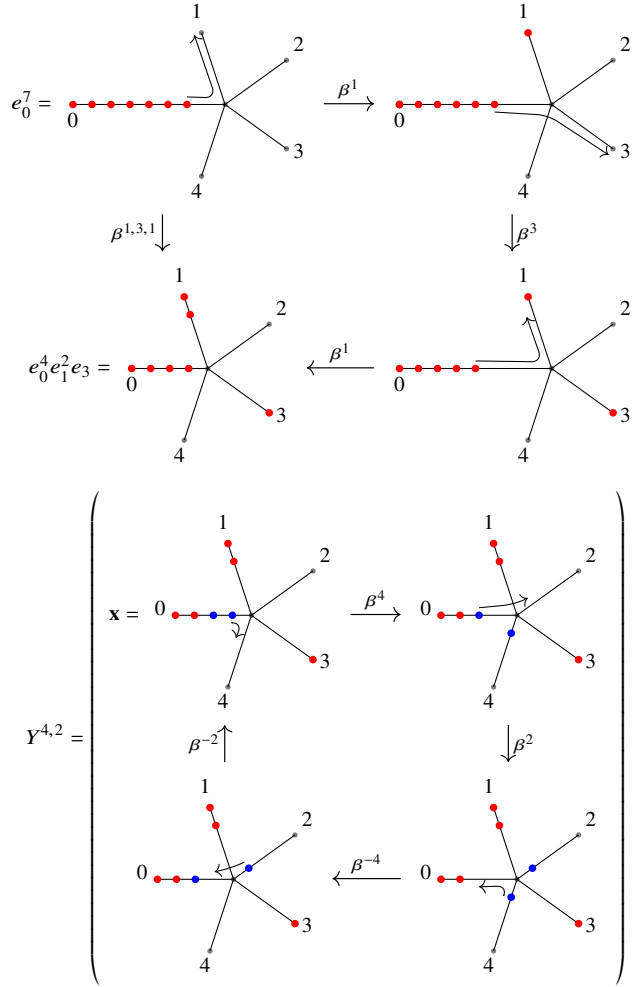

Roughly speaking, $\sigma_i^{\bfa,a,b}$ interchanges $i$th and $(i+1)$st particles by using $0$th, $a$th and $b$th edges after moving the first $(i-1)$ particles to leaves determined by sequence $\bfa$.
Since $\left(Y^{ab}\right)^{-1}=Y^{ba}$ and $Y^{aa}=1$, we have
\begin{align}
\left(\sigma_i^{\bfa,a,b}\right)^{-1}&=\sigma_i^{\bfa,b,a},&\sigma_i^{\bfa,a,b}&=1.
\end{align}
\begin{remark}
By counting the number of possible sequences $\bfa=(a_1,\hdots,a_{i+1})$ for $i=1,\hdots,n-1$, we get that the total number of $Y$-exchanges in $\bfB_n(S_k)$ is 
\[
\sum_{i=1}^{n-1}(k-1)^{i-1}\binom{k-1}2.
\]
Factors $(k-1)^{i-1}$ come from choosing $i-1$ from $k-1$ leaves to accommodate the first $i-1$ particles. Factors $\binom{k-1}2$ stem from the choice of two leaves where the exchange of particles $i$ and $i+1$ takes place.
\end{remark}
Next, we will proceed with the description of relations between the star-generators. To this end, we need to set up some more notation. Consider a path that connects configuration $e_0^n$ with configuration $e_0^{n-\sum b_j}e_1^{b_1}\cdots e_{k-1}^{b_{k-1}}$. Such a path will be unambiguously determined by the following sequence of nonnegative integers $\bfb=(b_1,\hdots, b_{k-1})$ that encodes the final particle configuration with $b_i$ particles on leaf $e_i$. Namely, we associate with $\bfb$ the following sequence
\[
\bar\bfb = (\underbrace{k-1,k-1,\hdots, k-1}_{b_{k-1}},\underbrace{k-2,k-2,\hdots, k-2}_{b_{k-2}},\hdots,\underbrace{1,1,\hdots, 1}_{b_1}).
\]
Path $\beta^{\bar\bfb}$ connects configuration $e_0^n$ with configuration $e_0^{n-\sum b_j}e_1^{b_1}\cdots e_{k-1}^{b_{k-1}}$ as desired. Moreover, when written as a word of $1$-cells of $D_n(\Gamma)$, $\beta^{\bar\bfb}$ consists only of collapsible cells. To see this, note first that all $1$-cells associated with one-particle moves in $\beta^{\bar\bfb}$, $c=\{e,v_1,\hdots,v_{n-1}\}$, are such that i) $e$ is order-respecting and ii) vertex $\iota(e)$ is the lowest unblocked vertex in the $0$-cell $c_0=\{\iota(e),v_1,\hdots,v_{n-1}\}$. This in turn means that $c=W_0\left(c_0\right)$, i.e. $c$ is collapsible. Using similar arguments, one can check that a critical cell $v_i(\bfb)$ for $\bfb=(b_1,\hdots,b_{k-1})$ corresponds to the following concatenation of paths
\begin{equation}\label{equation:critical cell interpretation}
v_i(\bfb) = F^\infty\left(\beta^{\bar{\bfb+e_i}}\cdot\beta^{-i}\cdot\beta^{-\bar\bfb}\right),
\end{equation}
where
\[
\bfb\pm e_i:=(b_1,\hdots, b_{i-1}, b_i\pm1, b_{i+1},\hdots, b_{k-1}).
\]
To see this, note that the only critical cell in the word corresponding to \ref{equation:critical cell interpretation} comes from the middle $\beta^{-i}$-move and this is exactly critical cell $v_i(\bfb)$.

\begin{proposition}\label{proposition:new generators}
The set $\{\sigma_i^\bfa\mid 1\le i\le n-1, \bfa=(a_1,\hdots, a_{i+1}), 1\le a_j\le k-1\}$ generates $\bfB_n(S_k)$.
\end{proposition}
\begin{proof}
Let us start with the simplest case of $n=2$. Any critical $1$-cell has the form $\{e_i,w\}$ where $e_i$ is the edge from $i$th leaf that is incident to the central vertex $v$ and $w$ is a vertex that is adjacent to $v$ and $v<w<\iota(e)$. According to our shorthand notation, such a $1$-cell is denoted by $v_i(w)$. Any two-particle exchange has the form $\sigma_1^{a,b}=Y^{a,b}$. It is straightforward to see that 
\[v_i(w)=F^\infty\left(Y^{\iota(e),w}\right).\]

In general, arbitrary critical cell $v_i(\bfb)\in D_n(\Gamma)$ can be expressed as a word of a number of star generators. In the remaining part of the proof, we will construct an inductive way for finding such a word. In order to handle arbitrary $n$, for any $\bfb=(b_1,\hdots,b_{k-1})$ consider the following splitting $\bfb_1:=(0,\hdots, 0,b_i, b_{i+1},\hdots, b_{k-1})$, $\bfb_2:=(b_1,\hdots, b_{i-1},0,\hdots, 0)$. From equation \eqref{equation:critical cell interpretation} we get that $v_i(\bfb)$ is the following conjugation
\begin{align*}
v_i(\bfb) &\sim_{F^\infty} \beta^{\bar{\bfb_1}} B  \beta^{-\bar{\bfb_1}},&
B&=\beta^i \beta^{\bar{\bfb_2}} \beta^{-i} \beta^{-\bar\bfb_2}.
\end{align*}
Ley us next show that $B$ is a word of $Y$-exchanges (modulo conjugation by appropriate paths) with the starting configuration being the end point of $\beta^{\bar\bfb_1}$, i.e.  $e_0^{n-b_i-\hdots-b_{k-1}}e_i^{b_i}\hdots e_{k-1}^{b_{k-1}}$. Indeed, for any sequence $\bfa=(a_1,\hdots,a_\ell)$, we have the following:
\[
\beta^i\beta^{\bfa}\beta^{-i}\beta^{-\bfa}=\beta^i\beta^{a_1}\beta^{\bfa'}\beta^{-i}\beta^{-\bfa'}\beta^{-a}=Y^{i,a_1}\cdot
\beta^{a_1}\left(
\beta^i\beta^{\bfa'}\beta^{-i}\beta^{-\bfa'}
\right)\beta^{-a_1},
\]
where $\bfa'=(a_2,\hdots, a_\ell)$. The above expression allows us to use the induction on the length of $\bfa$. Namely, by substituting $\bfa=\bar\bfb_2$, we have
\begin{gather*}
v_i(\bfb) \sim_{F^\infty} \left(\beta^{\bar\bfb_1}Y^{i,i-1}\beta^{-\bar\bfb_1}\right)\cdot \left(\beta^{\bar\bfb_1}\beta^{i-1}\left(\beta^i\beta^{\bar{\bfb'}}\beta^{-i}\beta^{-\bar{\bfb'}}\right)\beta^{-(i-1)}\beta^{-\bar\bfb_1}\right)= \\
=\sigma_{i_1}^{\bar\bfb_1,i,i-1}\cdot \left(\beta^{\bar\bfb_1}\beta^{i-1}\left(\beta^i\beta^{\bar{\bfb'}}\beta^{-i}\beta^{-\bar{\bfb'}}\right)\beta^{-(i-1)}\beta^{-\bar\bfb_1}\right),
\end{gather*}
where $\bar{\bfb'}=(b_1,\hdots,b_{i-1}-1)$ and $i_1=1+b_{k-1}+b_{k-2}+\hdots+b_i$. By iterating the above inductive step for $\beta^i\beta^{\bar{\bfb'}}\beta^{-i}\beta^{-\bar{\bfb'}}$, we get that $v_i(\bfb)$ can be expressed as the $F^\infty$-image of a concatenation of the resulting star generators.
\end{proof}

\begin{proposition}\label{proposition:new relations}
There are relations among $\sigma_i^\bfa$'s as follows:
\begin{enumerate}
\item Pseudo-commutative relation: for $|j-i|\ge 2$,
\begin{equation}\label{equation:pseudo-commutative}
\sigma_j^{a_1\hdots a_{j+1}}\sigma_i^{a_1\hdots a_{i+1}} = \sigma_i^{a_1\hdots a_{i+1}}\sigma_j^{a_1\hdots a_{i-1}a_{i+1}a_ia_{i+2}\hdots a_{j+1}}.
\end{equation}
\item Pseudo-braid relation: for $1\le i\le n-2$,
\begin{align}\label{equation:pseudo-braid}
&\mathrel{\hphantom{=}}\sigma_{i+1}^{a_1\hdots a_{i-1}a_i a_{i+1}a_{i+2}}\sigma_i^{a_1\hdots a_{i-1}a_i a_{i+2}}\sigma_{i+1}^{a_1\hdots a_{i-1}a_{i+2}a_i a_{i+1}}\notag\\
&=\sigma_i^{a_1\hdots a_{i-1}a_i a_{i+1}}\sigma_{i+1}^{a_1\hdots a_{i-1}a_{i+1}a_i a_{i+2}}\sigma_i^{a_1\hdots a_{i-1}a_{i+1}a_{i+2}}.
\end{align}
\end{enumerate}
\end{proposition}
\begin{proof}
The relations can be checked in a straightforward way by using the definition of star generators and expanding each generator as 
\[\sigma_i^{\bf b}=\beta^{b_1}\hdots\beta^{b_{i-1}}Y^{b_i,b_{i+1}}\beta^{-b_{i-1}}\hdots\beta^{-b_1},\]
where $Y^{b_i,b_{i+1}}=\beta^{b_i}\beta^{b_{i+1}}\beta^{-b_i}\beta^{-b_{i+1}}$ and cancelling $\beta^x\beta^{-x}$ whenever such an expression is encountered.
\end{proof}

Note that relation \eqref{equation:pseudo-commutative} becomes trivial if $a_i=a_{i+1}$ or $a_j=a_{j+1}$. Similarly, relations \eqref{equation:pseudo-braid} become trivial is at least one of the following equations is satisfied: $a_i=a_{i+1}$, $a_{i+1}=a_{i+2}$,  $a_i=a_{i+2}$. In particular, relations \eqref{equation:pseudo-braid} are nontrivial if and only if $k\ge 4$.

Let $G_{n,k}$ be the abstract group whose generators are all $\sigma_i^{\bfa}$ and relations are given by \eqref{equation:pseudo-commutative} and \eqref{equation:pseudo-braid}.

\begin{lemma}
Group $G_{n,k}$ is a free group of rank $N_1(n,k)$ that is generated by $\sigma_i^{\bfa}$ such that $1\leq a_1\le a_2\leq \hdots\le a_{i-1}\leq a_n$ and $a_{i}<a_{i+1}$.
\end{lemma}
\begin{proof}
We consider the group presentation of $G_{n,k}$ having all $\sigma_j^\bfa$ as generators and relations as in \eqref{equation:pseudo-commutative} and \eqref{equation:pseudo-braid}.

First, let us consider generators $\sigma_{j}^{\bfa}$ with $j=n-1$. Whenever $a_{i+1}<a_i$ for some $i\leq n-3$, we can use a preudo-commutative relation \eqref{equation:pseudo-commutative} to swap $a_{i+1}$ and $a_i$. That is, we apply the following Tietze transformation
\[\sigma_{n-1}^{a_1\hdots a_{n}}=\sigma_i^{a_1\hdots a_{i+1}}\sigma_{n-1}^{a_1\hdots a_{i-1}a_{i+1}a_ia_{i+2}\hdots a_{n}}\left(\sigma_i^{a_1\hdots a_{i+1}}\right)^{-1}\]
to obtain a presentation of $G_{n,k}$ where every $\sigma_{n-1}^{\bfa}$ satisfies $a_i\leq a_{i+1}$. Applying analogous Tietze transformations for every $i\in\{1,\hdots,n-3\}$, we get a presentation of $G_{n,k}$ where sequences $\bfa$ satisfy
\[a_1\le a_2\leq \hdots\le a_{n-2}.\]
Let us next show how pseudo-braiding relations \eqref{equation:pseudo-braid} allow us to arrange triples $a_{n-2}, a_{n-1}, a_n$. Whenever the triple $a_{n-2}, a_{n-1}, a_n$ consists of pairwise distinct integers, relation \eqref{equation:pseudo-braid} involves the following three generators
\[
\sigma_{n-1}^{a_1\hdots a_{n-3} a_{n-2} a_{n-1} a_n},\ 
\sigma_{n-1}^{a_1\hdots a_{n-3} a_{n-1} a_n a_{n-2}}\quad\text{ and }\quad
\sigma_{n-1}^{a_1\hdots a_{n-3} a_{n-1} a_n a_{n-2}}
\]
and two more generators with $j=n-2$. Note that the above three generators contain all permutations of $a_{n-2}, a_{n-1}, a_n$. Therefore, no matter what magnitudes of these numbers are, one can always find an appropriate Tietze transformation that yields a presentation of $G_{n,k}$ with 
\begin{align*}
a_{n-2}&\le a_n,&
a_{n-1}&<a_n.
\end{align*} 

Summing up, using relations \eqref{equation:pseudo-commutative} and \eqref{equation:pseudo-braid} we are able to obtain a presentation of $G_{n,k}$ whose generators $\sigma_{n-1}^{\bfa}$ are associated with sequences $\bfa$ such that $1\leq a_1\le a_2\leq \hdots\le a_{n-2}\leq a_n$ and $a_{n-1}<a_n$. Let $a_n=j$. Then the number of such sequences is $\binom{n+j-3}{j-1}(j-1)$, where the factor $(j-1)$ comes from the choice of $1\le a_{n-1}<j=a_n$.
Therefore the total number of generators $\sigma_{n-1}^\bfa$ is exactly
\[
\sum_{j=1}^{k-1}\binom{n+j-3}{j-1}(j-1).
\]
It is straightforward to check that this number is the same as $N_1(n,k)-N_1(n-1,k)$.

Finally, by applying analogous Tietze transformations for $\sigma_{n-2}^\bfa, \sigma_{n-3}^\bfa$ and so on, we can reduce all relations and end up with exactly $N_1(n,k)$ generators.
\end{proof}

\begin{theorem}
The abstract group $G_{n,k}$ is isomorphic to the braid group $\bfB_n(S_k)$ via the canonical map $\phi:G_{n,k}\to \bfB_n(S_k)$ which sends $\sigma_i^\bfa$ to an $n$-braid in $S_k$ represented by $\sigma_i^\bfa$.
\end{theorem}
\begin{proof}
By Propositions~\ref{proposition:new generators} and \ref{proposition:new relations}, map $\phi$ is a surjective group homomorphism.
Since both $G_{n,k}$ and $\bfB_n(S_k)$ are free groups of the same rank $N_1$, the surjective homomorphism $\phi$ becomes an isomorphism since any finitely generated free groups are Hopfian.
\end{proof}

The construction of star generators in terms of $Y$-exchanges can also be generalised so that it applies to any graph $\graf$. This is done by considering an appropriate embedding of $S_k$ into $\Gamma$ as explained in the definition below. 
\begin{definition}
For each essential vertex $v\in V(\graf)$ of valency $k\ge 3$ and a sequence $\bfa=(a_1,\hdots, a_{i+1})$ with $1\le a_j\le k-1$, we define $\sigma_i^{v;\bfa}$ by the image of $\sigma_i^{\bfa}\in\bfB_n(S_k)$ under the map $(\iota_v)_*$
\[
\sigma_i^{v;\bfa} = (\iota_v)_*(\sigma_i^{\bfa}),
\]
where $\iota_v$ is the order-preserving embedding
\[
\iota_v:(S_k,*)\to (T,*)\subset(\graf,*)
\]
such that $\iota_v$ maps the central vertex of $S_k$ to the chosen vertex $v$.
We also denote $\sigma_1^{v;a,b}$ by $Y^{v;ab}$.
\end{definition}

\subsection{Generators of loop type}\label{sec:loop-gen}
Generators of the loop type are in a one-to-one correspondence with deleted edges $e\in \graf\setminus T$. To each deleted edge $e$ we assign a unique loop in $\Gamma$ using the path in $T$ that joins $\iota(e)$ and $\tau(e)$. The form of such a loop is $e\cup [\iota(e),\tau(e)]$ and we give it an orientation that agrees with the canonical orientation of $e$ -- from $\iota(e)$ to $\tau(e)$. We also distinguish one deleted edge $e_0$ being the deleted edge which is adjacent to the root of $T$. Edge $e_0$ will define a move which will be the counterpart of the total braid $\delta$ in $\bfB_n(\RR^2)$. Other edges will correspond to one-particle moves. The base point for all moves is such that all particles are resting on the edge in $T$ which is incident  to the root $*$. In terms of the corresponding Morse complex $M_n(\Gamma)$, all loops are based at the point corresponding to the unique critical $0$-cell $\{0,1,\hdots,n-1\}$. The loop generators are constructed as follows (see Fig.\ref{fig:loop-gen} an illustration of these generators).
\begin{enumerate}
\item If $e=e_0$, its corresponding loop is the following embedding of circle graph $C_k$. 
\[
\iota_e:C_k=\begin{tikzpicture}[baseline=-.5ex,scale=0.5,xscale=-1]
\foreach \i in {-1,...,2} {
\draw({72*(\i+1)}:1) --({72*\i}:1);
\fill({72*\i}:1) circle (0.1);
}
\fill(216:1) node[below] {$*$} circle (0.1);
\draw[dashed](216:1)-- node[midway,below] {$e_0$} (288:1);
\end{tikzpicture}\to (T\cup e_0,*)\subset\graf
\]
Generator $\delta$ is the image of the generator of $\bfB_n(S^1)$ which takes the top particle around $C_k$ and moves it to the bottom of the line. Generator $\delta$ will be called the \emph{circular move}. Such a move is unique (up to homotopy) provided that edge $e_0$ exists.

\item If $e\neq e_0$, then there is an embedding $\iota_e:Q\to\graf$, where $Q$ is the lollipop graph
\[
\iota_e:Q=\begin{tikzpicture}[baseline=-.5ex,scale=0.5,xscale=-1]
\foreach \i in {-1,...,2} {
\draw({72*(\i+1)}:1) --({72*\i}:1);
\fill({72*\i}:1) circle (0.1);
}
\fill(216:1) circle (0.1) (2,0) circle (0.1);
\draw (1,0)--(2,0) node[left] {$*$};
\draw[dashed](216:1)-- node[midway,below] {$e$} (288:1);
\end{tikzpicture}\to(T\cup e,*)\subset \graf.
\]
In this case, generator $\gamma$ is the image of a particular generator of $\bfB_n(Q)$ which moves the top particle around $\gamma$ while the remaining $(n-1)$ particles stay fixed on their positions on the edge which is the stick of the lollipop. We call $\gamma$ a \emph{one-particle move}. Then there are exactly $(b_1(\Gamma)-1)$ such generators where $b_1(\Gamma)$ is the first Betti number of $\graf$.
\end{enumerate}


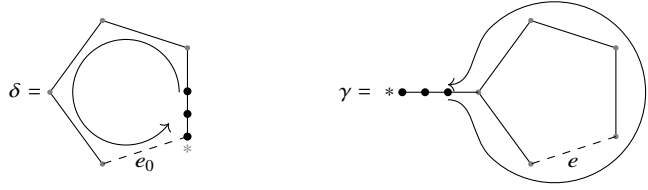
\begin{figure}[ht]
\begin{align*}
\delta&=\begin{tikzpicture}[baseline=-.5ex,scale=1,xscale=-1]
\foreach \i in {-1,...,2} {
\draw({72*(\i+1)}:1) --({72*\i}:1);
\fill[gray]({72*\i}:1) circle (1pt);
}
\fill[gray](216:1) node[below] {$*$} circle (0.05);
\draw[dashed](216:1)-- node[midway,below] {$e_0$} (288:1);
\fill (216:1) circle (1.5pt) ++(0,0.3) circle (1.5pt) ++(0,0.3) circle (1.5pt);
\draw[->] (180:0.7) arc (180:-144:0.7);
\end{tikzpicture}&
\gamma&=\begin{tikzpicture}[baseline=-.5ex,scale=1,xscale=-1]
\foreach \i in {-1,...,2} {
\draw({72*(\i+1)}:1) --({72*\i}:1);
\fill[gray]({72*\i}:1) circle (1pt);
}
\fill[gray](216:1) circle (1pt) (2,0) circle (1pt);
\draw (1,0)--(2,0) node[left] {$*$};
\draw[dashed](216:1)-- node[midway,below] {$e$} (288:1);
\fill (2,0) circle (1.5pt) ++(-0.3,0) circle (1.5pt) ++(-0.3,0) circle (1.5pt);
\draw[->] (1.4,-0.1) to[out=180,in=45] (-45:1.2) arc (-45:-315:1.2) to[out=-45,in=180] (1.4,0.1);
\end{tikzpicture}
\end{align*}
\caption{Generators of loop type}
\label{fig:loop-gen}
\end{figure}


Let us next define another one-particle move $\gamma_0$ for $\bfB_n(\graf)$ associated with edge $e_0$. Since $\graf$ is connected and contains at least one essential vertex, the loop represented by $e_0$ must contain an essential vertex $v$.
This means that we have a subgraph homeomorphic to the lollipop graph $Q$ whose loop is represented by $e_0$ and the lollipop's stick is some edge incident at $v$ that is not contained in the loop. We denote by $e_{0}^v$ and $e_{1}^v$ the two edges that are incident at $v$ and lie on the loop in $Q$. Edge $e_{0}^v$ is defined as the one which is closer to $*$ than $e_{1}^v$ in $T$. We denote the other edge incident at $v$ by $e_{2}^v$.
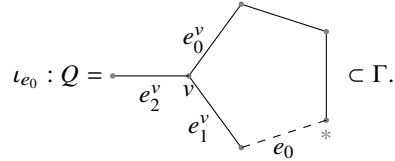
\[
\iota_{e_0}:Q=\begin{tikzpicture}[baseline=-.5ex,xscale=-1]
\draw(0:1) --(-72:1) node[midway, below left=-2pt] {$e_1^v$};
\fill[gray](-72:1) circle (1pt);
\draw(72:1) --(0:1) node[midway, left] {$e_{0}^v$};
\fill[gray](0:1) circle (1pt);
\foreach \i in {1,2} {
\draw({72*(\i+1)}:1) --({72*\i}:1);
\fill[gray]({72*\i}:1) circle (1pt);
}
\draw(1,0) node[below] {$v$} --(2,0) node[midway, below] {$e_{2}^v$};
\fill[gray](216:1) node[below] {$*$} circle (1pt) (2,0) circle (1pt);
\draw[dashed](216:1)-- node[midway,below] {$e_0$} (288:1);
\end{tikzpicture}\subset \graf.
\]

We define a braid $\gamma_0$ which is a one-particle move defined as follows. We first move the top $n-1$ particles to the edge $e_{2}^v$ and then we move the remaining $n$th particle along the loop represented by $e_0$. After the $n$th particle finishes the loop and returns to its original position, we move back the other $n-1$ particles from $e_{2}^v$ to their initial positions.
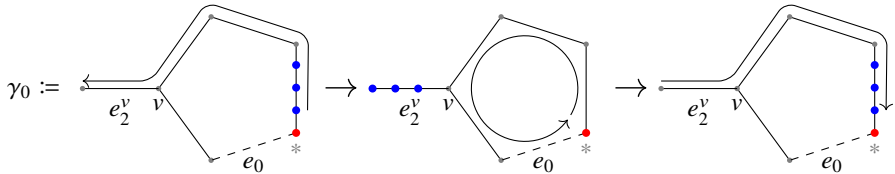
\[
\gamma_0\coloneqq\begin{tikzcd}[column sep=1pc]
\begin{tikzpicture}[baseline=-.5ex,xscale=-1]
\foreach \i in {-1,...,2} {
\draw({72*(\i+1)}:1) --({72*\i}:1);
\fill[gray]({72*\i}:1) circle (1pt);
}
\draw(1,0) node[below] {$v$} --(2,0) node[midway, below] {$e_{2}^v$};
\fill[gray](216:1) node[below] {$*$} circle (1pt) (2,0) circle (1pt);
\draw[dashed](216:1)-- node[midway,below] {$e_0$} (288:1);
\fill[red] (216:1) circle (1.5pt);
\fill[blue] (216:1) ++(0,0.3) circle (1.5pt) ++(0,0.3) circle (1.5pt) ++(0,0.3) circle (1.5pt);
\draw[rounded corners,->] (216:1)++(-0.15,0.3) -- (144:1.2) -- (72:1.2) -- (1.1,0.1) -- (2,0.1);
\end{tikzpicture}\arrow[r]&
\begin{tikzpicture}[baseline=-.5ex,xscale=-1]
\foreach \i in {-1,...,2} {
\draw({72*(\i+1)}:1) --({72*\i}:1);
\fill[gray]({72*\i}:1) circle (1pt);
}
\draw(1,0) node[below] {$v$} --(2,0) node[midway, below] {$e_{2}^v$};
\fill[gray](216:1) node[below] {$*$} circle (1pt) (2,0) circle (1pt);
\draw[dashed](216:1)-- node[midway,below] {$e_0$} (288:1);
\fill[red] (216:1) circle (1.5pt);
\fill[blue] (2,0) circle (1.5pt) (1.7,0) circle (1.5pt) (1.4,0) circle (1.5pt);
\draw[->] (210:0.7) arc (210:-144:0.7);
\end{tikzpicture}\arrow[r]&
\begin{tikzpicture}[baseline=-.5ex,xscale=-1]
\foreach \i in {-1,...,2} {
\draw({72*(\i+1)}:1) --({72*\i}:1);
\fill[gray]({72*\i}:1) circle (1pt);
}
\draw(1,0) node[below] {$v$} --(2,0) node[midway, below] {$e_{2}^v$};
\fill[gray](216:1) node[below] {$*$} circle (1pt) (2,0) circle (1pt);
\draw[dashed](216:1)-- node[midway,below] {$e_0$} (288:1);
\fill[red] (216:1) circle (1.5pt);
\fill[blue] (216:1) ++(0,0.3) circle (1.5pt) ++(0,0.3) circle (1.5pt) ++(0,0.3) circle (1.5pt);
\draw[rounded corners,<-] (216:1)++(-0.15,0.3) -- (144:1.2) -- (72:1.2) -- (1.1,0.1) -- (2,0.1);
\end{tikzpicture}
\end{tikzcd}
\]
Then $\gamma_0$ can be expressed as a word involving $\sigma^{v;\bfa}$, $\gamma_0$ and $\delta$. More specifically, $Y$-exchanges $\sigma_i^{v;\bfa}$ at $v$ satisfy the following relation
\begin{equation}\label{eq:delta-rel}
\gamma_0=\sigma_{n-1}^{v;2,\hdots,2,1}\sigma_{n-2}^{v;2,\hdots,2,1}\cdots\sigma_1^{v;2,1}\delta.
\end{equation}
Moreover, for any $\bfa=(a_1,\hdots,a_{i+1})$ such that $a_i\in\{1,2\}$ for all $i$, we have another relation
\begin{align}\label{eq:lollipop-braid}
\sigma_{i+1}^{v;\bfa'}&=\delta\cdot\sigma_i^{v;\bfa}\cdot\delta^{-1},&\bfa'&=(1,a_1,\hdots,a_{i+1}).
\end{align}
Both of the above relations can be verified by visualising their corresponding moves. LHS and RHS of the above equations are essentially the same moves up to some backtracking.

\begin{remark}
Let us next briefly argue how the above relations \eqref{eq:delta-rel} and \eqref{eq:lollipop-braid} will be used to relate graph braid groups to $\bfB_n(\RR^2)$. Namely, consider quotient group $\bfB_n(Q)/\langle\gamma_0\rangle$ which is generated only by star generators $\sigma_i^{v;\bfa}$. By \eqref{eq:delta-rel}, circular move is then expressed as
\[\delta=\sigma_1^{v;1,2}\sigma_2^{v;2,1,2}\hdots\sigma_{n-1}^{v;2,\hdots,1,2}.\]
Substituting the above expression in equation \eqref{eq:lollipop-braid}, we get
\begin{align}\label{equation:braid-like relation}
\sigma_1^{v;1,2}\sigma_2^{v;2,1,2}X\sigma_1^{v;1,2}X^{-1} &= \sigma_2^{v;1,1,2}\sigma_1^{v;1,2}\sigma_2^{v;2,1,2},
\end{align}
where
\begin{align*}
X&=\sigma_3^{v;2,2,1,2}\sigma_4^{v;2,\hdots,2,1,2}\cdots\sigma_{n-1}^{v;2,\hdots,2,1,2}.
\end{align*}
If we forget all decorations such as $v$ and $\bfa$, then we can use pseudo-braiding relations for $S_3$ to show that $X$ commutes with $\sigma_1^{v;1,2}$. Then, equation \eqref{equation:braid-like relation} simplifies to the ordinary braid relation.
\end{remark}

\subsection{Generation of graph braid groups and connectivity}\label{sec:2connected}
The following theorem summarises the material contained in previous sections.
\begin{theorem}\label{thm:generators}
The graph braid group $\bfB_n(\graf)$ for arbitrary $\Gamma$ is generated by the following moves
\begin{enumerate}
\item $Y$-exchanges 
\[
\{\sigma_i^{v;\bfa}\mid 1\le i<n, v\in V, \bfa=(a_1,\hdots, a_{i+1}), 1\le a_j<\val(v)\},
\]
\item $b_1(\graf)$-many one-particle moves
\[
\{\gamma_0,\hdots, \gamma_{b-1}\mid b=b_1(\graf)\},
\]
\item one circular move $\delta$,
\end{enumerate}
where $b_1(\graf)$ is the first Betti number of $\graf$.
\end{theorem}
\begin{proof}
This is a direct consequence of Farley-Sabalka's algorithm and Proposition~\ref{proposition:new generators}. Namely, any critical $1$-cell associated with an essential vertex $v\in\Gamma$ is $F^\infty$-equivalent to a word of $Y$-exchanges $\sigma_i^{v;\bfa}$. Furthermore, any critical $1$-cell associated with a deleted edge $e\neq e_0$ is $F^\infty$-equivalent to a word of $1$-cells corresponding to the loop generator associated with $e$. This is because such a word contains only collapsible cells and one critical cell which is precisely $\{e,0,1,\hdots,n-2\}$. By similar arguments, critical cell $\{e_0,1,2,\hdots,n-1\}$ is $F^\infty$-equivalent to the word of $1$-cells associated with circular move $\delta$. Hence, for every critical $1$-cell in $D_n(\Gamma)$, we have found its corresponding geometric generator. By theorem \ref{thm:FS}, such geometric moves generate $\bfB_n(\graf)$.
\end{proof}

Our next result states that for $2$-connected graphs the number of $Y$-exchanges needed to generate $\bfB_n(\Gamma)$ can be greatly reduced.
\begin{proposition}\label{proposition:generators}
Let $\graf$ be a $2$-connected graph. Then the braid group $\bfB_n(\graf)$ is generated by 
\[
\{Y^{v;ab}\mid v\in V, 1\le a<b<\val(v)\}\sqcup
\{\gamma_0,\hdots,\gamma_{b-1}\mid b=b_1(\graf)\}\sqcup\{\delta\}.
\]
\end{proposition}
\begin{proof}
It suffices to show that each $\sigma_i^{v;\bfa}$ can be expressed as a word of $Y^{w;ab}$, $\gamma_j$, and $\delta$. We use the induction on both $i$ and the edge-length of the path $[*,v]$ in $T$. For $i=1$, there is nothing to prove for any $v$.

Suppose that $v$ is the very first essential vertex when counting from $*$, i.e. there are no essential vertices in the interior of $[*,v]$. Let $v'$ be the essential vertex which is the end of $a_1$th edge incident at $v$. Since $\graf$ is $2$-connected, $\graf_v=\graf\setminus\st(v)$ is connected.
Therefore there exists a path $\gamma\subset \graf_v$ from $v'$ to $*$. Such a path $\gamma$ defines a loop $\bar\gamma$ based at $*$ which is the union $\bar\gamma\coloneqq[*,v']\cup\gamma$.
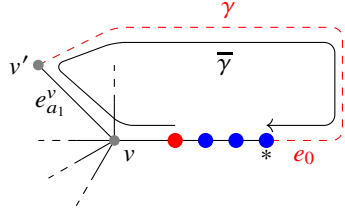
\[
\begin{tikzpicture}
\draw(0,0) node[below right] {$v$} -- (2,0) node[below] {$*$};
\draw(0,0)-- node[midway, left] {$e_{a_1}^v$} (-1,1) node[left] {$v'$};
\draw(0,0) -- (90:0.5) (0,0)--(180:0.5) (0,0)--(210:0.5) (0,0)--(240:0.5);
\draw[dashed] (90:0.5)--(90:1) (180:0.5)--(180:1) (210:0.5)--(210:1) (240:0.5)--(240:1);
\draw[red,dashed, rounded corners] (-1,1) -- (0,1.5) --node[midway, above] {$\gamma$} (3,1.5) -- (3,0) -- node[below] {$e_0$} (2,0);
\fill[gray] (0,0) circle (2pt) (2,0) circle (2pt) (-1,1) circle (2pt);
\fill[blue] (2,0) circle (0.1) (1.6,0) circle (0.1) (1.2,0) circle (0.1);
\fill[red] (0.8,0) circle (0.1);
\draw[rounded corners, ->] (0.8,0.2) -- (0.1,0.2) -- (-0.8, 1) -- (0,1.3) -- node[midway, below] {$\bar\gamma$} (2.9,1.3) -- (2.9,0.2) -- (2,0.2);
\end{tikzpicture}
\]

We use loop $\bar\gamma$ to define a braid of the circular type which sends the top particle around $\bar\gamma$ to the last position, just as the circular move $\delta$. We denote such a move by $\bar\gamma$ as well. Move $\bar\gamma$ is a word of $\gamma_j^{\pm1}$ and $\delta$ so that
\[
\bar\gamma=\gamma_{j_1}^{s_1}\hdots\gamma_{j_\ell}^{s_\ell}\delta,
\]
where $(j_1,\hdots,j_\ell)$ is a sequence of indices of deleted edges that $\gamma$ passes and $s_k\in\{-1,1\}$, $1\leq k\leq \ell$, are exponents coming from orientations of the deleted edges relative to the orientation of $\gamma$. We also have a relation which is analogous to relation \eqref{eq:lollipop-braid}. 
\begin{align*}
\sigma_i^{v;\bfa} &= \bar\gamma\sigma_{i-1}^{v;\bfa'}\left(\bar\gamma\right)^{-1},&
\bfa'&=(a_2,\hdots, a_{i+1})
\end{align*}
By repeating the above inductive step for $\sigma_{i-1}^{v;\bfa'}$ and another loop associated with edge $e_{a_2}^v$ and so on, we end up with an expression that involves only one-particle moves, the circular move $\delta$ and $\sigma_1^{v;a_{i},a_{i+1}}$.

Now suppose that  $v$ is a vertex farther than the nearest vertices from $*$ and $i\ge2$. 
As before, the connectivity of $\graf_v$ implies the existence of a path $\gamma\subset\graf_v$ from the $v'$ to $*$, where $v'\neq v$ is the vertex which is the end point of the $a_1$th edge adjacent to $v$.

Suppose that $\gamma$ does not intersect the path $[*,v]$ in $T$.
Then the union $[*,v']\cup\gamma$ defines a loop $\bar\gamma$ based at $*$ and by the exactly same argument as above, we have
\begin{equation}\label{2conn-case1}
\sigma_i^{v;\bfa} = \bar\gamma\sigma_{i-1}^{v;\bfa'}\left(\bar\gamma\right)^{-1},
\end{equation}
where 
\begin{align*}
\bfa'&=(a_2,\hdots, a_{i+1}),&
\bar\gamma&=\gamma_{j_1}^{s_1}\hdots\gamma_{j_\ell}^{s_\ell}\delta,
\end{align*}
and by the induction hypothesis, we are done.
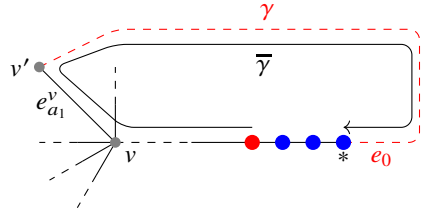
\[
\begin{tikzpicture}
\draw[dashed](0,0) node[below right] {$v$} -- (1,0);
\draw (1,0) -- (3,0) node[below] {$*$};
\draw(0,0)-- node[midway, left] {$e_{a_1}^v$} (-1,1) node[left] {$v'$};
\draw(0,0) -- (90:0.5) (0,0)--(180:0.5) (0,0)--(210:0.5) (0,0)--(240:0.5);
\draw[dashed] (90:0.5)--(90:1) (180:0.5)--(180:1) (210:0.5)--(210:1) (240:0.5)--(240:1);
\draw[red,dashed, rounded corners] (-1,1) -- (0,1.5) --node[midway, above] {$\gamma$} (4,1.5) -- (4,0) -- node[below] {$e_0$} (3,0);
\fill[gray] (0,0) circle (2pt) (3,0) circle (2pt) (-1,1) circle (2pt);
\fill[blue] (3,0) circle (0.1) (2.6,0) circle (0.1) (2.2,0) circle (0.1);
\fill[red] (1.8,0) circle (0.1);
\draw[rounded corners, ->] (1.8,0.2) -- (0.1,0.2) -- (-0.8, 1) -- (0,1.3) -- node[midway, below] {$\bar\gamma$} (3.9,1.3) -- (3.9,0.2) -- (3,0.2);
\end{tikzpicture}
\]

Finally, suppose that $\gamma$ intersects the path $[*,v']$ at $w\neq *$.
By taking the subpath of $\gamma$, we may assume that $\gamma$ is a composition of $\gamma'$ and $[w,*]$ such that $\gamma'$ is a path joining $v$ and $w$ and does not intersect $[*,v]$.
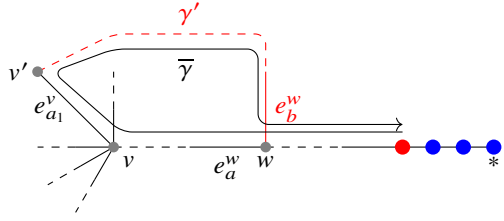
\[
\begin{tikzpicture}
\draw[dashed](0,0) node[below right] {$v$} -- (1,0);
\draw[dashed](2,0) node[below] {$w$} -- (3,0);
\draw (1,0)-- node[midway,below] {$e_a^w$}(2,0);
\draw(3,0) -- (5,0) node[below] {$*$};
\draw(0,0)-- node[midway, left] {$e_{a_1}^v$} (-1,1) node[left] {$v'$};
\draw(0,0) -- (90:0.5) (0,0)--(180:0.5) (0,0)--(210:0.5) (0,0)--(240:0.5);
\draw[dashed] (90:0.5)--(90:1) (180:0.5)--(180:1) (210:0.5)--(210:1) (240:0.5)--(240:1);
\draw[red,dashed, rounded corners] (-1,1) -- (0,1.5) --node[midway, above] {$\gamma'$} (2,1.5) -- (2,1);
\draw[red] (2,1)-- node[midway,right] {$e_b^w$} (2,0);
\fill[gray] (0,0) circle (2pt) (2,0) circle (2pt) (-1,1) circle (2pt);
\fill[blue] (5,0) circle (0.1) (4.6,0) circle (0.1) (4.2,0) circle (0.1);
\fill[red] (3.8,0) circle (0.1);
\draw[rounded corners, ->] (3.8,0.2) -- (0.1,0.2) -- (-0.8, 1) -- (0,1.3) -- node[midway, below] {$\bar\gamma$} (1.9,1.3) -- (1.9,0.3) -- (3.8,0.3);
\end{tikzpicture}
\]

As before, it gives us loop $\bar\gamma=[*,v']\cup \gamma'\cup [w,*]$ based at $*$.
In this case, $\bar\gamma$ can be regarded as a one-particle move and expressed as a word
\[
\bar\gamma=\gamma_{j_1}^{s_1}\hdots\gamma_{j_\ell}^{s_\ell},
\]
where $(j_1,\hdots,j_\ell)$ is the sequence of indices of deleted edges defined as the same as before and $s_k\in\{-1,1\}$, $1\leq k\leq \ell$ are the exponents.

Let us denote two edges adjacent to $w$ contained in $[*,v]$ and $\gamma'$ by $e_a^w$ and $e_b^w$, respectively.
The conjugation $\bar\gamma^{-1}\sigma_i^{v;\bfa}\bar\gamma$ gives us a braid, which is a concatenation as follows:
\begin{enumerate}
\item Move the first particle along $[*,w]$ to the edge $e_b^w$.
\item Move the next $(i-2)$ particles onto edges adjacent to $v$ by using the sequence $(a_2,\hdots, a_{i-1})$.
\item Interchange the positions of the next two particles by using the $a_i$th and $a_{i+1}$th edges adjacent to $v$ and take them back to the edge at $*$.
\item Move the $(i-2)$ particles on edges of $v$ back to the original position.
\item Move the first particle back to the original position.
\end{enumerate}

Then indeed, this is a conjugate of $\sigma_{i-1}^{v;\bfa'}$ with $\bfa'=(a_2,\hdots, a_{i+1})$ by the braid $\sigma$ which interchanges positions of the first particle with the next $i$ particles at $w$ by using edges $e_a^w$ and $e_b^w$, which is a word of $Y$-exchanges 
\[
\sigma=\sigma_1^{w;b,a}\sigma_2^{w;a,b,a}\hdots\sigma_{i}^{w;a,\hdots,a,b,a}.
\]
Therefore we have
\begin{equation}\label{2conn-case2}
\sigma_i^{v;\bfa} = \bar\gamma\cdot \sigma\cdot \sigma_{i-1}^{v;\bfa'}\cdot \sigma^{-1}\cdot\bar\gamma^{-1}
\end{equation}

By induction hypothesis, not only $\sigma_{i-1}^{v;\bfa'}$ but also all $\sigma_j^{w;-}$ can be expressed as words of $\sigma_i^{w;-}$'s since $w$ is closer to $*$ than $v$, which completes the proof.
\end{proof}


\section{Planar triconnected graphs}\label{sec:3connected}
For the rest of the paper, we assume that $\graf$ is planar and we fix its planar embedding $\iota:\graf\to\RR^2$. We denote the closures of connected components of the complement $\RR^2\setminus\graf$ by  $D_0,\hdots,D_b$, where $b$ is the first betti nuber of $\graf$.
In particular, we assume that $D_b$ is the unbounded component.
\begin{remark}\label{remark:regions}
If $\graf$ is $2$-connected, then each bounded $D_i$ is homeomorphic to a closed disk and will be called a face of $\Gamma$.
\end{remark}

We choose a bounded face, say $D_0$, whose boundary $\partial D_0$ shares at least one edge $e$ with the unbounded face $\partial D_b$. By subdividing $e$ if necessary, we may assume that there exists a bivalent vertex $*$ in the common boundary of $D_0$ and $D_b$
\begin{align*}
*&\in\partial D_0\cap \partial D_b,&
\val(*)&=2.
\end{align*}
We denote the facial cycle $\partial D_0$ by $\delta$.

Let $(T,*)$ be a rooted spanning tree for $\graf$ such that every deleted edge has bivalent endpoints. Without loss of generality, we may assume that $\delta\setminus e_0\subset T$, or equivalently, $e_0\subset \graf\setminus T$ is a deleted edge representing the loop $\delta$.

\begin{proposition}\label{braid-2connected}
Let $\graf$ be a planar $2$-connected graph. Then $\bfB_n(\graf)$ admits a group presentation generated by $Y$-exchanges, one-particle moves and the circular move such that we obtain the classical braid group $\bfB_n(\RR^2)$ from $\bfB_n(\graf)$ by 
\begin{enumerate}
\item taking the quotient by all one-particle moves $\{\gamma_0,\hdots,\gamma_{b-1}\}$, and 
\item identifying all $Y$-exchanges.
\end{enumerate}
\end{proposition}
\begin{proof}
Let $(T,*)$ be a rooted spanning tree as above.
Then $\bfB_n(\graf)$ admits a group presentation whose generators are $Y$-exchanges, one-particle moves and a circular move by Proposition~\ref{proposition:generators}.

It is clear that the induced map $\iota_*:\bfB_n(\graf)\to\bfB_n(\RR^2)$ kills all one-particle moves and identifies all $Y^{v;a,b}$'s. Therefore $\iota_*$ factors through $\bfB_n(\graf)/\sim$
\[
\begin{tikzcd}
\bfB_n(\graf)\arrow[rr,"\iota_*"]\arrow[rd,"\psi"] & & \bfB_n(\RR^2)\\
&\bfB_n(\graf)/\sim\arrow[ru,"\bar\iota_*"]
\end{tikzcd}
\]
where
\[
\bfB_n(\graf)/\sim=\bfB_n(\graf)\bigg/\left\langle
\gamma_0,\hdots,\gamma_{b-1}, Y^{v;a,b}Y^{v';b',a'}
\right\rangle
\]
for $1\le a<b<\val(v), 1\le a'<b'<\val(v')$ and $v, v'\in V$. Taking the quotient $\sim$ is the same as forgetting decorations $\bfa$ and $v$ in $\sigma_i^{v;\bfa}$. Here we only sketch the proof of this fact, as it relies on the same inductive reasoning as the one that was used in the proof of Proposition \ref{proposition:generators}. Namely, for an essential vertex $v$ consider the inductive step where one assumes that $\sigma_j^{w;\bfa}\sim\sigma_j^{w';\bfb}$ for all $j<i$ and $w,w'<v$ and any sequences $\bfa,\ \bfb$. Recursive relations \eqref{2conn-case1} and \eqref{2conn-case2} together with braiding and commutative relations for each $w,w'<v$ imply that  $\sigma_i^{v;\bfa}\sim\sigma_i^{w;\bfb}$ for all $w<v$ and and appropriate $\bfb$. The base case is $\sigma_1^{v_0;1,2}$  where $v_0$ is the closest essential vertex to the root $*$.

Hence, pseudo-commutative star relations \eqref{equation:pseudo-commutative} become just the commutative relations in $\bfB_n(\graf)/\sim$ and lollipop relations \eqref{equation:braid-like relation} become the regular $2D$ braiding relations. Therefore, we have 
\begin{align*}
\bar\sigma_i\bar\sigma_j\bar\sigma_i^{-1}\bar\sigma_j^{-1}=\begin{cases}
e & |i-j|>1;\\
\bar\sigma_j^{-1}\bar\sigma_i& |i-j|=1,
\end{cases}
\end{align*}
where $\bar\sigma_i\coloneqq\psi(\delta^{i-1}Y^{v;a,b}\delta^{1-i})$ for any $v$ and $1\le a<b<\val(v)$.

In other words, there exists a surjective homomorphism $f:\bfB_n(\RR^2)\to\bfB_n(\graf)/\sim$ sending $\sigma_i$ to $\bar\sigma_i$.
Since $\bar\iota_*\circ f$ is the identity on $\bfB_n(\RR^2)$, the map $f$ is injective as well. Therefore $f$ is in fact an isomorphism.
\end{proof}

\begin{example}\label{example:theta}
Consider graph $\Theta_3$ from Fig. \ref{theta-moves}. Denote by $\gamma$ the one-particle move where the first particle travels around the cycle which is disjoint from $*$ in the direction from $v$ to $w$ along $[v,w]$ (it is the upper cycle on Fig. \ref{theta-moves}). Moreover, denote by $\gamma'$ a move which involves particles $1$ and $2$ where i) particle $1$ travels to branch $e_2^w$ through the solid edge $[v,w]$, ii) particle $2$ goes around the upper cycle (first along $[v,w]$ in the direction from $v$ to $w$ and then back to $v$ through the upper deleted edge) and iii) the first particle goes back through the solid edge $[v,w]$.
Up to some backtracking moves, we have the following relations.
\begin{equation}\label{theta-rel1-1}
Y^{w;2,1}\gamma=\gamma' Y^{v; 2,1},
\end{equation}
\begin{equation}\label{theta-rel1-2}
\delta\gamma=\gamma' \delta.
\end{equation}
Note, that after substituting expression for $\gamma'$ from (\ref{theta-rel1-1}) into (\ref{theta-rel1-2}) we obtain 
\begin{equation}\label{theta-rel}
\delta\gamma\delta^{-1}=Y^{w;2,1}\gamma Y^{v;1,2}.
\end{equation}
After taking the quotient by one-particle move $\gamma$, equation \eqref{theta-rel} yields identification $Y^{w;1,2}\sim Y^{v; 1,2}$. Hence, by proposition \ref{braid-2connected}, group $\bfB_n(\Theta_3)$ becomes $\bfB_n(\RR^2)$ already after taking the quotient by one-particle cycles.
\end{example}

\begin{example}\label{example:vartheta}
Let $\vartheta_3$ be the graph which is the union of the theta graph $\Theta_3$ and an edge as depicted in Figure~\ref{vartheta-moves}.
We denote by $\gamma$ and $\gamma'$ the one-particle moves corresponding to the loops which are boundary of the right and left regions, respectively.
Then as before, up to some backtracking moves, we have the following relation.
\begin{equation}\label{vartheta-rel}
Y^{w;1,2}=\gamma' Y^{v;1,2} \gamma Y^{z;2,1} \gamma'^{-1} Y^{z;1,2} \gamma.
\end{equation}
After taking the quotient by $\gamma, \gamma'$, this yields identification $Y^{w;1,2}\sim Y^{v;1,2}$.
\end{example}

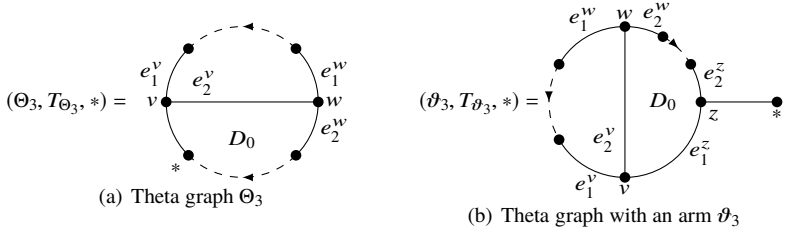
\begin{figure}[h]
\subfigure[Theta graph $\Theta_3$\label{theta-moves}]{\makebox[0.45\textwidth]{$
(\Theta_3, T_{\Theta_3},*)=\begin{tikzpicture}[baseline=-.5ex,scale=0.5]
\draw(45:2) arc (45:-45:2) (135:2) arc (135:225:2);
\draw(-2,0) -- (2,0) (-1,0) node[above] {$e_2^v$};
\begin{scope}[dashed,decoration={
    markings,
    mark=at position 0.5 with {\arrow{latex}}}
    ] 
    \draw[postaction={decorate}](45:2) arc (45:135:2);
    \draw[postaction={decorate}](-45:2) arc (315:225:2);
\end{scope}
\fill(225:2) node[below left] {$*$};
\fill(-45:2) circle (4pt) (0:2) node[right] {$w$} circle (4pt) (45:2) circle (4pt);
\fill(-22:-2) node[left] {$e_1^v$} (-45:-2) circle (4pt) (0:-2) node[left] {$v$} circle (4pt) (45:-2) circle (4pt) (-22:2) node[right] {$e_2^w$} (22:2) node[right] {$e_1^w$};
\draw(0,-1) node {$D_0$};
\end{tikzpicture}
$}}
\subfigure[Theta graph with an arm $\vartheta_3$\label{vartheta-moves}]{\makebox[0.45\textwidth]{$
(\vartheta_3, T_{\vartheta_3},*)=\begin{tikzpicture}[baseline=-.5ex]
\draw (30:1) arc (30:-150:1) node[pos=0.05,right] {$e_2^z$} node[pos=0.4, right] {$e_1^z$} (60:1) arc (60:150:1) node[pos=0.05, above] {$e_2^w$} node[pos=0.7, above] {$e_1^w$};
\begin{scope}[dashed,decoration={
    markings,
    mark=at position 0.5 with {\arrow{latex}}}
    ] 
    \draw[postaction={decorate}](60:1) arc (60:30:1);
    \draw[postaction={decorate}](150:1) arc (150:210:1);
\end{scope}
\draw (0,1) node[above] {$w$}--(0,-1) node[below] {$v$};
\draw (-120:1) node[below] {$e_1^v$} (0,-0.5) node[left] {$e_2^v$};
\draw (1,0) node[below right] {$z$} --(2,0);
\draw (0.5,0) node {$D_0$};
\fill (0,1) circle (2pt) (0,-1) circle (2pt) (1,0) circle (2pt) (30:1) circle (2pt) (60:1) circle (2pt) (150:1) circle (2pt) (210:1) circle (2pt) (2,0) circle (2pt);
\draw (2,0) node[below] {$*$};
\end{tikzpicture}
$}}
\caption{Settings for canonical moves on $\Theta_3$ and $\vartheta_3$.}
\label{fig:thetas}
\end{figure}

\begin{definition}
Let $p, q$ be two essential vertices of valency $k, \ell$.
For $1\le a<b<k$ and $1\le c<d<\ell$, two triples $(p;a,b)$ and $(q;c,d)$ are said to be \emph{$\Theta$-related} if there exists and embedding $\iota$
\[
\iota:(\Theta_3, T_{\Theta_3},*)\to (\graf, T,*)\quad\text{ or }\quad
\iota:(\vartheta_3, T_{\vartheta_3},*)\to (\graf, T,*)
\]
such that 
\begin{align*}
\iota(v) &= p,&
\iota(e_1^v) &= e_a^p,&
\iota(e_2^v) &= e_b^p,\\
\iota(w) &= q,&
\iota(e_1^w) &= e_c^q,&
\iota(e_2^w) &= e_d^q.
\end{align*}
We define the equivalence relation generated by $\Theta$-relatedness and denote it by
\[
(p;a,b)\sim_\Theta (q;c,d).
\]
\end{definition}

Suppose that $(p;a,b)$ and $(q;c,d)$ are $\Theta$-related.
Then two $Y$-exchanges $Y^{p;ab}$ and $Y^{q;cd}$ are the images of $Y$-exchanges $Y^{v;1,2}$ and $Y^{w;1,2}$ under the induced map 
\[
\iota_*:\bfB_n(\Theta_3)\to \bfB_n(\graf)\quad\text{ or }\quad
\iota_*:\bfB_n(\vartheta_3)\to \bfB_n(\graf).
\]
Moreover, the induced map $\iota_*$ sends one-particle moves in $\bfB_n(\Theta_3)$ or $\bfB_n(\vartheta_3)$ to words of one-particle moves in $\bfB_n(\graf)$. It also maps a circular move in $\bfB_n(\Theta_3)$ to a word of one-particle moves and a circular move in $\bfB_n(\graf)$ as seen in the proof of Proposition~\ref{proposition:generators}.

\begin{lemma}\label{lemma:related}
If $(p;a,b)$ and $(q;c,d)$ are $\Theta$-related, then $Y$-exchanges $Y^{p;ab}$ and $Y^{q;cd}$ are identified by taking the quotient by one-particle moves.
\end{lemma}
\begin{proof}
Let $\iota$ be an embedding of $\Theta_3$ or $\vartheta_3$ which makes $(p;a,b)$ and $(q;c,d)$ $\Theta$-related.
As seen in Examples~\ref{example:theta} and \ref{example:vartheta}, two $Y$-exchanges $Y^{v;1,2}$ and $Y^{w;1,2}$ will be identified in the quotient by one-particle moves.
Since both $Y^{p;ab}$ and $Y^{q;cd}$ are images of $Y^{v;1,2}$ and $Y^{w;1,2}$ under $\iota_*$ and $\iota_*$ maps one-particle moves on $\Theta_3$ or  $\vartheta_3$ to words of one-particle moves on $\Gamma$, we are done.
\end{proof}

\begin{proposition}\label{proposition:all related}
Suppose that $\graf$ is $3$-connected.
Then each pair of triples $(p;a,b)$ and $(q;c,d)$ are $\Theta$-related.
\end{proposition}
We will prove this proposition later.

\begin{proof}[Proof of Theorem~\ref{main theorem}]
Let us fix a rooted spanning tree $(T,*)$ as before.
We consider the quotient group
\[
\bar\bfB_n(\graf)=\bfB_n(\graf)/\langle\gamma_0,\hdots, \gamma_{b-1}\rangle
\]
by all one-particle moves.

The fact that all triples $(p;a,b)$ are $\Theta$-related (Proposition~\ref{proposition:all related}) implies that all $Y$-exchanges are identified in the quotient group $\bar\bfB_n(\graf)$ by Lemma~\ref{lemma:related}. Finally, Proposition~\ref{braid-2connected} implies that the quotient group is indeed the classical braid group $\bfB_n(\RR^2)$.
\end{proof}
The whole technical difficulty of the proof lies in the proof of Proposition~\ref{proposition:all related}. This is the subject of the next subsection. 

\subsection{Proof of Proposition~\ref{proposition:all related}}
From now on, we assume that $\graf$ is $3$-connected. We start by invoking the following graph-theoretic theorem that can be found in \cite{BondyMurty}. 
\begin{theorem}\label{nonseparating}
A cycle in a simple $3$-connected planar graph is a facial cycle if and only if it is nonseparating.
\end{theorem}
Recall that a separating cycle C is a cycle for which $\Gamma\setminus C$ is disconnected and a facial cycle is the boundary of a face. This in particular ensures the connectivity of $\graf\setminus \delta$ for $\delta=\partial D_0$.

Let $v\in\delta$ be an essential vertex of valency $k$. By the labelling convention, the edge $e_{k-1}^v$ is on $\delta$.

\begin{lemma}\label{lemma:special case}
Let $p\neq q \in \delta$ be essential vertices of valency $k,\ell$.
For each $1\le a<k$ and $1\le b<\ell$, two $Y$-exchange $(p;a,k-1)$ and $(q;b,\ell-1)$ are $\Theta$-related.
\end{lemma}
\begin{proof}
Since $\graf$ is $3$-connected, the complement of $\delta$ is connected. Therefore there is a path $\gamma$ joining $p$ and $q$ whose starting and ending edges are $e_a^p$ and $e_b^q$. Hence we have the following embedding of $(\Theta_3,T_{\Theta_3},*)$
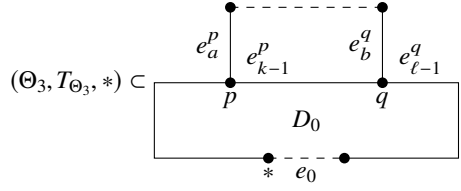
\[
(\Theta_3,T_{\Theta_3},*)\subset
\begin{tikzpicture}[baseline=-.5ex,yshift=-1cm]
\draw(-0.5,0) -- (-2,0) -- (-2,1) -- (-1,1) -- (0,1) node[midway,above] {$e_{k-1}^p$}-- (1,1) -- (2,1) node[midway,above] {$e_{\ell-1}^q$}-- (2,0) -- (0.5,0) (-1,1)--(-1,2) node[midway,left] {$e_a^p$} (1,1)--(1,2) node[midway,left] {$e_b^q$};
\draw[dashed] (-0.5,0)--node[below] {$e_0$}(0.5,0) (-1,2)--(1,2);
\draw(0,0.5) node {$D_0$};
\fill(-0.5,0) circle(2pt) node[below] {$*$} (0.5,0) circle (2pt);
\fill(-1,1) circle (2pt) node[below] {$p$} (1,1) circle (2pt) node[below] {$q$} (-1,2) circle(2pt) (1,2) circle(2pt);
\end{tikzpicture}
\]
and so two triples $(p;a,k-1)$ and $(q;b,\ell-1)$ are $\Theta$-related.
\end{proof}

Let $(p;a,b)$ be a triple for an essential vertex $p$ and $1\le a<b<k=\val(p)$.
We denote two endpoints of $e_a^p$ and $e_b^p$ other than $p$ by $p_a$ and $p_b$, respectively.
Let $\gamma_a$ and $\gamma_b$ be the paths in $T$ from $p_a$ and $p_b$ to the loop $\delta$.
We denote the intersection of $\gamma_a$ and $\gamma_b$ by $\gamma$.
\begin{align*}
\gamma&=\gamma_a\cap\gamma_b,&
\gamma_a&=\gamma\cup e_a^p,&
\gamma_b&=\gamma\cup e_b^p.
\end{align*}
We orient $\gamma_a$ and $\gamma_b$ from $\delta$ to $p_a$ and $p_b$, respectively, and so we may think two sides of $\gamma_a$ and $\gamma_b$. Next, we define two subsets $L$ and $R$ of $\RR^2$ by the following unions of faces $D\neq D_0$.
\begin{enumerate}
\item Regions $D$ in $L$ intersects $\gamma_a\setminus\{p_a\}$ and is on the left side of $\gamma_a$.
\item Regions $D$ in $R$ intersects $\gamma_b\setminus\{p_b\}$ and is on the right side of $\gamma_b$.
\end{enumerate}

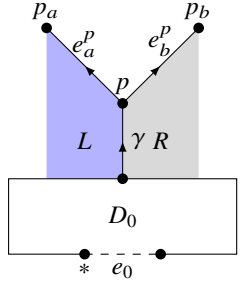
\begin{align*}
\begin{tikzpicture}[baseline=-.5ex,yshift=-1cm]
\fill[blue, opacity=0.3] (-1,1)--(0,1)--(0,2)--(-1,3)--cycle;
\fill[gray, opacity=0.3] (1,1)--(0,1)--(0,2)--(1,3)--cycle;
\draw(0.5,1.5) node {$R$} (-0.5,1.5) node {$L$};
\draw(-0.5,0) -- (-1.5,0) -- (-1.5,1) -- (1.5,1) -- (1.5,0) -- (0.5,0);
\draw[dashed] (-0.5,0)--node[below] {$e_0$}(0.5,0);
\draw(0,0.5) node {$D_0$};
\begin{scope}[decoration={
    markings,
    mark=at position 0.5 with {\arrow{latex}}}
    ] 
\draw[postaction={decorate}] (0,1)-- (0,2) node[midway, right] {$\gamma$};
\draw[postaction={decorate}] (0,2)-- (-1,3) node[midway,above] {$e_a^p$};
\draw[postaction={decorate}] (0,2)-- (1,3) node[midway,above] {$e_b^p$};
\end{scope}
\fill(-0.5,0) circle(2pt) node[below] {$*$} (0.5,0) circle (2pt);
\fill (0,1) circle (2pt) (0,2) circle (2pt) node[above] {$p$} (-1,3) circle (2pt) node[above] {$p_a$} (1,3) circle (2pt) node[above] {$p_b$};
\end{tikzpicture}
\end{align*}
We claim that regions $L$ and $R$ are interior-disjoint. To see this, note that if is a face $D\subset L\cap R$, then $D$ cannot be homeomorphic to a disk. This in turn contradicts the fact that $\graf$ is $3$-connected. See Remark~\ref{remark:regions}.

\begin{lemma}\label{lemma:general case}
The triple $(p;a,b)$ is $\Theta$-related with $(q;c,\ell-1)$ for some $q\in\delta$ and $1\le c<\ell=\val(q)$.
\end{lemma}
\begin{proof}
We first suppose that $L\cap R=\gamma$. Then we claim that the triple $(p;a,b)$ is $\Theta$-related with $(q;c,\ell-1)$ for some $q\in\delta$ and $1\le c<\ell=\val(q)$.

According to whether $L$ or $R$ contains the unbounded region $D_b$, we have one of the following three cases.
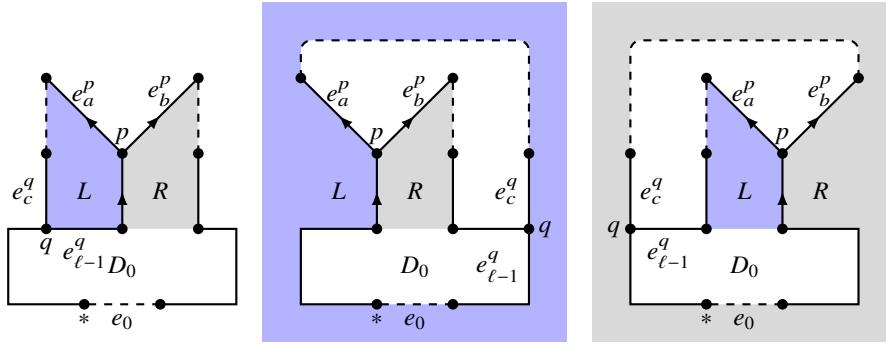
\[
\begin{tikzpicture}[baseline=-.5ex,yshift=-1cm]
\fill[blue, opacity=0.3] (-1,1)--(0,1)--(0,2)--(-1,3)--cycle;
\fill[gray, opacity=0.3] (1,1)--(0,1)--(0,2)--(1,3)--cycle;
\draw(0.5,1.5) node {$R$} (-0.5,1.5) node {$L$};
\draw[thick](-0.5,0) -- (-1.5,0) -- (-1.5,1) -- (-1,1) -- (0,1) node[midway,below] {$e_{\ell-1}^q$};
\draw[dashed,thick] (0.5,0)--node[below] {$e_0$}(-0.5,0);
\draw(0,0.5) node {$D_0$};
\begin{scope}[thick, decoration={
    markings,
    mark=at position 0.5 with {\arrow{latex}}}
    ] 
\draw[postaction={decorate}] (0,1)-- (0,2);
\draw[postaction={decorate}] (0,2) node[above] {$p$}-- (-1,3) node[midway,above] {$e_a^p$};
\draw[postaction={decorate}] (0,2)-- (1,3) node[midway,above] {$e_b^p$};
\end{scope}
\draw[thick](-1,1) node[below] {$q$}--(-1,2) node[midway,left] {$e_c^q$} (1,2)--(1,1)--(1.5,1)--(1.5,0)--(0.5,0);
\draw[thick,dashed] (-1,2)--(-1,3) (1,2)--(1,3);
\fill(-1,1) circle (2pt) (1,1) circle (2pt) (-1,2) circle (2pt) (1,2) circle (2pt);
\fill(-0.5,0) circle(2pt) node[below] {$*$} (0.5,0) circle (2pt);
\fill (0,1) circle (2pt) (0,2) circle (2pt) (-1,3) circle (2pt) (1,3) circle (2pt);
\end{tikzpicture}\quad
\begin{tikzpicture}[baseline=-.5ex,yshift=-1cm]
\fill[blue, opacity=0.3] (-1.5,-0.5)rectangle (2.5,4);
\fill[white] (0,1)--(0,2)--(-1,3) {[rounded corners] --(-1,3.5)--(2,3.5)}--(2,0)--(-1,0)--(-1,1)--cycle;
\fill[gray, opacity=0.3] (1,1)--(0,1)--(0,2)--(1,3)--cycle;
\draw(0.5,1.5) node {$R$} (-0.5,1.5) node {$L$};
\draw[thick](0,0) -- (-1,0) -- (-1,1) -- (0,1);
\draw[dashed,thick] (1,0)--node[below] {$e_0$}(0,0);
\draw(0.5,0.5) node {$D_0$};
\begin{scope}[thick, decoration={
    markings,
    mark=at position 0.5 with {\arrow{latex}}}
    ] 
\draw[postaction={decorate}] (0,1)-- (0,2);
\draw[postaction={decorate}] (0,2) node[above] {$p$}-- (-1,3) node[midway,above] {$e_a^p$};
\draw[postaction={decorate}] (0,2)-- (1,3) node[midway,above] {$e_b^p$};
\end{scope}
\draw[thick] (1,2)--(1,1)--(2,1) node[right] {$q$}--(2,0) node[midway,left] {$e_{\ell-1}^q$}--(1,0) (2,2)--(2,1) node[midway,left] {$e_c^q$};
\draw[thick,dashed,rounded corners] (-1,3) -- (-1,3.5)-- (2,3.5)--(2,2);
\draw[thick,dashed] (1,2)--(1,3);
\fill(2,1) circle (2pt) (1,1) circle (2pt) (1,2) circle (2pt) (2,2) circle (2pt);
\fill(0,0) circle(2pt) node[below] {$*$} (1,0) circle (2pt);
\fill (0,1) circle (2pt) (0,2) circle (2pt) (-1,3) circle (2pt) (1,3) circle (2pt);
\end{tikzpicture}\quad
\begin{tikzpicture}[baseline=-.5ex,yshift=-1cm]
\fill[gray, opacity=0.3] (-2.5,-0.5)rectangle (1.5,4);
\fill[white] (0,1)--(0,2)--(1,3) {[rounded corners] --(1,3.5)--(-2,3.5)}--(-2,0)--(1,0)--(1,1)--cycle;
\fill[blue, opacity=0.3] (-1,1)--(0,1)--(0,2)--(-1,3)--cycle;
\draw(0.5,1.5) node {$R$} (-0.5,1.5) node {$L$};
\draw[thick](-1,0) -- (-2,0) -- (-2,1)--(-1,1) node[midway,below] {$e_{\ell-1}^q$} (0,1) -- (1,1)--(1,0)--(0,0);
\draw[dashed,thick] (0,0)--node[below] {$e_0$}(-1,0);
\draw(-0.5,0.5) node {$D_0$};
\begin{scope}[thick, decoration={
    markings,
    mark=at position 0.5 with {\arrow{latex}}}
    ] 
\draw[postaction={decorate}] (0,1)-- (0,2);
\draw[postaction={decorate}] (0,2) node[above] {$p$}-- (-1,3) node[midway,above] {$e_a^p$};
\draw[postaction={decorate}] (0,2)-- (1,3) node[midway,above] {$e_b^p$};
\end{scope}
\draw[thick] (-1,1)--(-1,2) (-2,1) node[left] {$q$}--(-2,2) node[midway,right] {$e_c^q$};
\draw[thick,dashed,rounded corners] (1,3) -- (1,3.5)-- (-2,3.5)--(-2,2);
\draw[thick,dashed] (-1,2)--(-1,3);
\fill(-2,1) circle (2pt) (-1,1) circle (2pt) (-1,2) circle (2pt) (-2,2) circle (2pt);
\fill(-1,0) circle(2pt) node[below] {$*$} (0,0) circle (2pt);
\fill (0,1) circle (2pt) (0,2) circle (2pt) (-1,3) circle (2pt) (1,3) circle (2pt);
\end{tikzpicture}
\]
It is easy to see that in each of the above cases  $(\graf,T,*)$ contains $(\Theta_3, T_{\Theta_3},*)$, which yields the $\Theta$-relatedness of $(p;a,b)$ and $(q;c,\ell-1)$.

Suppose next that $L$ and $R$ intersect at an essential vertex $r$ outside of $\gamma$. We use the induction on the length $N$ of the path $\gamma$. Since the vertex $*$ faces the unbounded region $D_b$, the only possibility for the shape of $L$ and $R$ looks as follows.
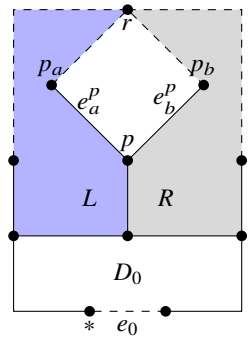
\[
\begin{tikzpicture}[baseline=-.5ex,yshift=-3cm]
\fill[blue,opacity=0.3] (0,1)--(0,2)--(-1,3)--(0,4)--(-1.5,4)--(-1.5,1)--cycle;
\fill[gray,opacity=0.3] (0,1)--(0,2)--(1,3)--(0,4)--(1.5,4)--(1.5,1)--cycle;
\draw(-0.5,1.5) node {$L$} (0.5,1.5) node {$R$};
\draw (-0.5,0)--(-1.5,0)--(-1.5,1)--(1.5,1)--(1.5,0)--(0.5,0);
\draw[dashed] (0,4)--(-1.5,4)--(-1.5,2) (0,4)--(1.5,4)--(1.5,2) (0.5,0)--(-0.5,0) node[midway,below] {$e_0$};
\draw (-1.5,1)--(-1.5,2) (1.5,1)--(1.5,2) (0,1)--(0,2)--(-1,3) node[midway,above] {$e_a^p$} (0,2)--(1,3) node[midway,above] {$e_b^p$};
\draw[dashed](-1,3)--(0,4)--(1,3);
\draw(0,0.5) node {$D_0$};
\fill (-0.5,0) circle (2pt) node[below] {$*$} (0.5,0) circle (2pt) (-1.5,2) circle (2pt) (1.5,2) circle (2pt) (-1.5,1) circle (2pt) (1.5,1) circle (2pt) (0,1) circle (2pt) (0,2) circle (2pt) node[above] {$p$} (0,4) circle (2pt) node[below] {$r$} (-1,3) circle (2pt) node[above] {$p_a$} (1,3) circle (2pt) node[above] {$p_b$};
\end{tikzpicture}
\]
\begin{remark}
It is also possible that $L$ and $R$ share an edge outside $\gamma$, but in this case the proof proceeds without changes.
\end{remark}

Since $\graf$ is $3$-connected, the complement of $\{p,r\}$ is connected as well. Hence, there exists a path in $\graf_{p,q}:=\graf\setminus\{p,r\}$ connecting any two regions of $\graf_{p,q}$. In particular, there exists a path that joins two connected components of $\Gamma_{p,q}\setminus(L\cup R)$ being the region bounded by the interior boundary of $L\cup R$ (the white diamond region on the picture above) and $D_0$. As such a path cannot touch $p$ and $q$, we claim that it necessarily has a common part with $\gamma$. Denote by $q\in\gamma$, $q\neq p$ the essential vertex where the above described path joins with $\gamma$. If the path has no common part with $\gamma$, then we contradict the construction of $L$ or $R$ by creating some additional faces of $\Gamma$ that imply that regions $L$ and $R$  do not intersect outside $\gamma$, which is a contradiction. This situation is depicted on figures below.
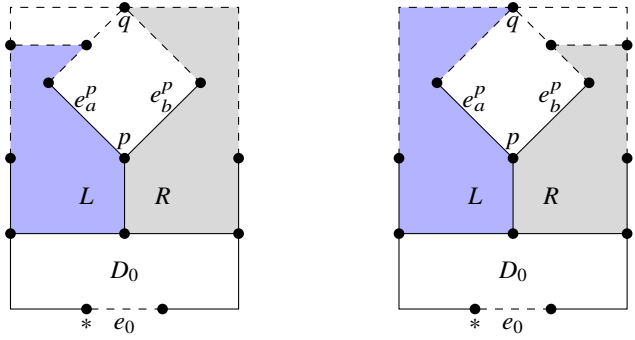
\begin{align*}
&\begin{tikzpicture}[baseline=-.5ex,yshift=-2cm]
\fill[blue,opacity=0.3] (0,1)--(0,2)--(-1,3)--(-0.5,3.5)--(-1.5,3.5)--(-1.5,1)--cycle;
\fill[gray,opacity=0.3] (0,1)--(0,2)--(1,3)--(0,4)--(1.5,4)--(1.5,1)--cycle;
\draw(-0.5,1.5) node {$L$} (0.5,1.5) node {$R$};
\draw (-0.5,0)--(-1.5,0)--(-1.5,1)--(1.5,1)--(1.5,0)--(0.5,0);
\draw[dashed] (0,4)--(-1.5,4)--(-1.5,2) (0,4)--(1.5,4)--(1.5,2) (0.5,0)--(-0.5,0) node[midway,below] {$e_0$};
\draw (-1.5,1)--(-1.5,2) (1.5,1)--(1.5,2) (0,1)--(0,2)--(-1,3) node[midway,above] {$e_a^p$} (0,2)--(1,3) node[midway,above] {$e_b^p$};
\draw[dashed](-1,3)--(0,4)--(1,3) (-0.5,3.5)--(-1.5,3.5);
\fill (-0.5,3.5) circle (2pt) (-1.5,3.5) circle (2pt);
\draw(0,0.5) node {$D_0$};
\fill (-0.5,0) circle (2pt) node[below] {$*$} (0.5,0) circle (2pt) (-1.5,2) circle (2pt) (1.5,2) circle (2pt) (-1.5,1) circle (2pt) (1.5,1) circle (2pt) (0,1) circle (2pt) (0,2) circle (2pt) node[above] {$p$} (0,4) circle (2pt) node[below] {$q$} (-1,3) circle (2pt) (1,3) circle (2pt);
\end{tikzpicture}&
&\begin{tikzpicture}[baseline=-.5ex,yshift=-2cm]
\fill[blue,opacity=0.3] (0,1)--(0,2)--(-1,3)--(0,4)--(-1.5,4)--(-1.5,1)--cycle;
\fill[gray,opacity=0.3] (0,1)--(0,2)--(1,3)--(0.5,3.5)--(1.5,3.5)--(1.5,1)--cycle;
\draw(-0.5,1.5) node {$L$} (0.5,1.5) node {$R$};
\draw (-0.5,0)--(-1.5,0)--(-1.5,1)--(1.5,1)--(1.5,0)--(0.5,0);
\draw[dashed] (0,4)--(-1.5,4)--(-1.5,2) (0,4)--(1.5,4)--(1.5,2) (0.5,0)--(-0.5,0) node[midway,below] {$e_0$};
\draw (-1.5,1)--(-1.5,2) (1.5,1)--(1.5,2) (0,1)--(0,2)--(-1,3) node[midway,above] {$e_a^p$} (0,2)--(1,3) node[midway,above] {$e_b^p$};
\draw[dashed](0.5,3.5)--(1.5,3.5);
\fill (0.5,3.5) circle (2pt) (1.5,3.5) circle (2pt);
\draw[dashed](-1,3)--(0,4)--(1,3);
\draw(0,0.5) node {$D_0$};
\fill (-0.5,0) circle (2pt) node[below] {$*$} (0.5,0) circle (2pt) (-1.5,2) circle (2pt) (1.5,2) circle (2pt) (-1.5,1) circle (2pt) (1.5,1) circle (2pt) (0,1) circle (2pt) (0,2) circle (2pt) node[above] {$p$} (0,4) circle (2pt) node[below] {$q$} (-1,3) circle (2pt) (1,3) circle (2pt);
\end{tikzpicture}
\end{align*}
Therefore we have the following situation:
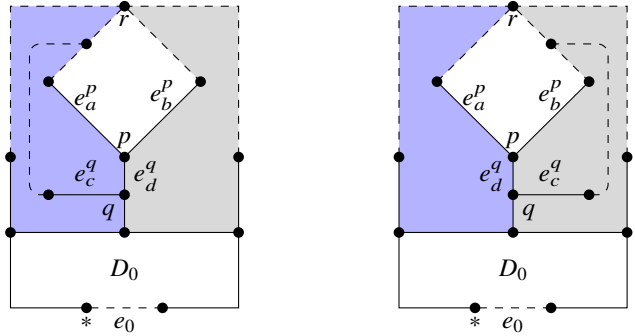
\begin{align*}
&\begin{tikzpicture}[baseline=-.5ex,yshift=-2cm]
\fill[blue,opacity=0.3] (0,1)--(0,2)--(-1,3)--(0,4)--(-1.5,4)--(-1.5,1)--cycle;
\fill[gray,opacity=0.3] (0,1)--(0,2)--(1,3)--(0,4)--(1.5,4)--(1.5,1)--cycle;
\draw (-0.5,0)--(-1.5,0)--(-1.5,1)--(1.5,1)--(1.5,0)--(0.5,0);
\draw[dashed] (0,4)--(-1.5,4)--(-1.5,2) (0,4)--(1.5,4)--(1.5,2) (0.5,0)--(-0.5,0) node[midway,below] {$e_0$};
\draw (-1.5,1)--(-1.5,2) (1.5,1)--(1.5,2) (0,1)--(0,2)--(-1,3) node[midway,above] {$e_a^p$} (0,2)--(1,3) node[midway,above] {$e_b^p$};
\draw[dashed](-1,3)--(0,4)--(1,3);
\draw[dashed,rounded corners](-0.5,3.5)--(-1.25,3.5)--(-1.25,1.5)--(-1, 1.5);
\draw (-1,1.5)--(0,1.5) node[midway,above] {$e_c^q$} (0,1.75) node[right] {$e_d^q$};
\fill (-0.5,3.5) circle (2pt) (-1,1.5) circle (2pt) (0,1.5) circle (2pt) node[below left] {$q$};
\draw(0,0.5) node {$D_0$};
\fill (-0.5,0) circle (2pt) node[below] {$*$} (0.5,0) circle (2pt) (-1.5,2) circle (2pt) (1.5,2) circle (2pt) (-1.5,1) circle (2pt) (1.5,1) circle (2pt) (0,1) circle (2pt) (0,2) circle (2pt) node[above] {$p$} (0,4) circle (2pt) node[below] {$r$} (-1,3) circle (2pt) (1,3) circle (2pt);
\end{tikzpicture}&
&\begin{tikzpicture}[baseline=-.5ex,yshift=-2cm]
\fill[blue,opacity=0.3] (0,1)--(0,2)--(-1,3)--(0,4)--(-1.5,4)--(-1.5,1)--cycle;
\fill[gray,opacity=0.3] (0,1)--(0,2)--(1,3)--(0,4)--(1.5,4)--(1.5,1)--cycle;
\draw (-0.5,0)--(-1.5,0)--(-1.5,1)--(1.5,1)--(1.5,0)--(0.5,0);
\draw[dashed] (0,4)--(-1.5,4)--(-1.5,2) (0,4)--(1.5,4)--(1.5,2) (0.5,0)--(-0.5,0) node[midway,below] {$e_0$};
\draw (-1.5,1)--(-1.5,2) (1.5,1)--(1.5,2) (0,1)--(0,2)--(-1,3) node[midway,above] {$e_a^p$} (0,2)--(1,3) node[midway,above] {$e_b^p$};
\draw[dashed,rounded corners](0.5,3.5)--(1.25,3.5)--(1.25,1.5)--(1, 1.5);
\draw (1,1.5)--(0,1.5) node[midway,above] {$e_c^q$} (0,1.75) node[left] {$e_d^q$};
\fill (0.5,3.5) circle (2pt) (1,1.5) circle (2pt) (0,1.5) circle (2pt) node[below right] {$q$};
\draw[dashed](-1,3)--(0,4)--(1,3);
\draw(0,0.5) node {$D_0$};
\fill (-0.5,0) circle (2pt) node[below] {$*$} (0.5,0) circle (2pt) (-1.5,2) circle (2pt) (1.5,2) circle (2pt) (-1.5,1) circle (2pt) (1.5,1) circle (2pt) (0,1) circle (2pt) (0,2) circle (2pt) node[above] {$p$} (0,4) circle (2pt) node[below] {$r$} (-1,3) circle (2pt) (1,3) circle (2pt);
\end{tikzpicture}
\end{align*}
The base case for induction is when $N=0$. This can happen only when, $p\in\delta$, which means that $\gamma=\{p\}$ and therefore there is no room for a vertex $q\in \gamma\setminus\{p\}$. We end up with the case which has already been dealt with at the beginning of this proof. 

If $N>0$, then one can find an embedded $(\vartheta_3,T_{\vartheta_3},*)$ in $(\graf,T,*)$ which makes $(p;a,b)$ and $(q;c,d)$ $\Theta$-related.
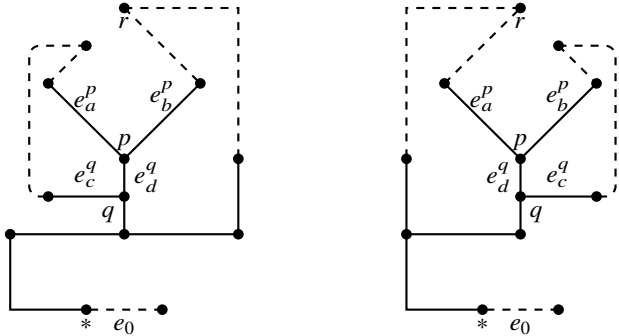
\begin{align*}
&\begin{tikzpicture}[baseline=-.5ex,yshift=-2cm]
\draw[thick] (-0.5,0)--(-1.5,0)--(-1.5,1)--(1.5,1)--(1.5,2);
\draw[thick,dashed] (0,4)--(1.5,4)--(1.5,2) (0.5,0)--(-0.5,0) node[midway,below] {$e_0$};
\draw[thick] (0,1)--(0,2)--(-1,3) node[midway,above] {$e_a^p$} (0,2)--(1,3) node[midway,above] {$e_b^p$};
\draw[thick,dashed](-1,3)--(-0.5,3.5) (0,4)--(1,3);
\draw[thick,dashed,rounded corners](-0.5,3.5)--(-1.25,3.5)--(-1.25,1.5)--(-1, 1.5);
\draw[thick] (-1,1.5)--(0,1.5) node[midway,above] {$e_c^q$} (0,1.75) node[right] {$e_d^q$};
\fill (-0.5,3.5) circle (2pt) (-1,1.5) circle (2pt) (0,1.5) circle (2pt) node[below left] {$q$};
\fill (-0.5,0) circle (2pt) node[below] {$*$} (0.5,0) circle (2pt) (1.5,2) circle (2pt) (-1.5,1) circle (2pt) (1.5,1) circle (2pt) (0,1) circle (2pt) (0,2) circle (2pt) node[above] {$p$} (0,4) circle (2pt) node[below] {$r$} (-1,3) circle (2pt) (1,3) circle (2pt);
\end{tikzpicture}&
&\begin{tikzpicture}[baseline=-.5ex,yshift=-2cm]
\draw[thick] (-0.5,0)--(-1.5,0)--(-1.5,1)--(0,1);
\draw[thick,dashed] (0,4)--(-1.5,4)--(-1.5,2) (0.5,0)--(-0.5,0) node[midway,below] {$e_0$};
\draw[thick] (-1.5,1)--(-1.5,2) (0,1)--(0,2)--(-1,3) node[midway,above] {$e_a^p$} (0,2)--(1,3) node[midway,above] {$e_b^p$};
\draw[thick,dashed,rounded corners](0.5,3.5)--(1.25,3.5)--(1.25,1.5)--(1, 1.5);
\draw[thick] (1,1.5)--(0,1.5) node[midway,above] {$e_c^q$} (0,1.75) node[left] {$e_d^q$};
\fill (0.5,3.5) circle (2pt) (1,1.5) circle (2pt) (0,1.5) circle (2pt) node[below right] {$q$};
\draw[thick,dashed](-1,3)--(0,4) (0.5,3.5)--(1,3);
\fill (-0.5,0) circle (2pt) node[below] {$*$} (0.5,0) circle (2pt) (-1.5,2) circle (2pt) (-1.5,1) circle (2pt) (0,1) circle (2pt) (0,2) circle (2pt) node[above] {$p$} (0,4) circle (2pt) node[below] {$r$} (-1,3) circle (2pt) (1,3) circle (2pt);
\end{tikzpicture}
\end{align*}
Clearly, the length from $q$ to $\delta$ is strictly shorter than $N$ and by the induction hypothesis, $(q;c,d)$ is $\Theta$-related with $(q';c',\ell'-1)$ for some $q'\in\delta$ and $1\le c'<\ell'=\val(q')$. Therefore, we have
\[
(p;a,b)\sim_\Theta(q;c,d)\sim_\Theta(q';c',\ell'-1)
\]
which completes the inductive step.
\end{proof}

\begin{proof}[Proof of Proposition~\ref{proposition:all related}]
By Lemma~\ref{lemma:general case}, any triple $(p;a,b)$ for $1\le a<b<k=\val(p)$ is $\Theta$-related with $(q;c,\ell-1)$ for some $q\in \delta$ and $1\le c<\ell=\val(q)$.
Hence it suffices to prove that all $(p;a,k-1)$ are $\Theta$-related for $p\in\delta$ and $1\le a<k=\val(p)$.

Now we apply Lemma~\ref{lemma:special case}. Since $\graf$ is $3$-connected, $\delta$ contains at least two essential vertices $p\neq q$.
Then all triples $(r;c,m-1)$ for $p\neq r\in\delta$ and $1\le c<m=\val(r)$ are $\Theta$-related with a single triple $(p;1,k-1)$ for $k=\val(p)$
\[
(r;c,m-1)\sim_\Theta(p;1,k-1).
\]

On the other hand, all $(p;a;k-1)$ for $1\le a<k$ are $\Theta$-related with a single triple $(q;1,\ell-1)$ for $\ell=\val(q)$ 
\[
(p;a,k-1)\sim_\Theta(q;1,\ell-1)\sim_\Theta(p;1,k-1)
\]
and so we are done.
\end{proof}

\section{A generalisation of the Yang-Baxter equation}\label{sec:yang-baxter}
As we have seen in previous sections, the quotient 
\[
\bar\bfB_n(\graf):=\bfB_n(\graf)/\langle\gamma_0,\hdots, \gamma_{b-1}\rangle
\]
for triconnected $\graf$ is isomorphic to the well-known $\bfB_n(\RR^2)$. In other words, group $\bar\bfB_n(\graf)$ can be generated by an appropriate set of $Y$-exchanges $\bar \sigma_1,\hdots,\bar\sigma_{n-1}$ that satisfy the braiding and commutative relations. However, when $\Gamma$ is only $2$-connected, the quotient group $\bar\bfB_n(\graf)$ generally has more generators than $\bfB_n(\RR^2)$ and relations between these generators can be more complicated than braiding and commutative relations. Our general strategy for a $2$-connected $\Gamma$ is to find a set of $Y$-exchanges associated with triples $\{(v_\mu,k_\mu,l_\mu)\}_{\mu=1}^r$ where $1\leq k_\mu<l_\mu<{\mathrm{val}}(v_\mu)$ such that the set
\[\left\{\sigma_1^{v_\mu;k_\mu,l_\mu}\mid\ 1\leq \mu\leq r\right\}\sqcup \{\delta\}\]
generates $\bar\bfB_n(\graf)$ and within a fixed triple the braiding and commutative relations are satisfied. What leads to generalisations of the Yang-Baxter equation are additional relations between generators corresponding to different triples $(v_\mu,k_\mu,l_\mu)$ that appear in the presentation of $\bar\bfB_n(\graf)$. By Proposition \ref{braid-2connected} we know that such relations become the ordinary braiding relations only when one identifies $\sigma_1^{v_\mu;k_\mu,l_\mu}\sim \sigma_1^{v_{\mu'},k_{\mu'},l_{\mu'}}$ for all $\mu,\mu'$.

In the remaining part of this section, we derive an explicit form of such a relation for graph $\Theta_4$.
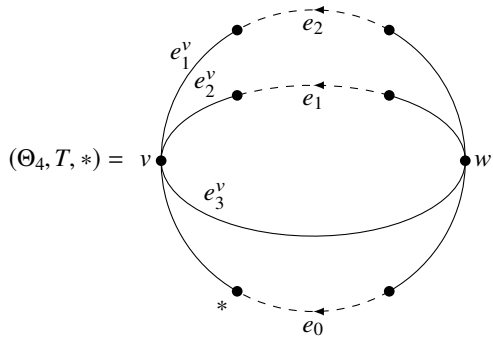
\begin{align*}
(\Theta_4,T,*)=\begin{tikzpicture}[baseline=-.5ex]
\fill(-2,0) circle (2pt) node[left] {$v$} (2,0) circle (2pt) node[right] {$w$} (60:2) circle (2pt) (150:2) node[above=3pt] {$e_1^v$} (-60:2) circle (2pt) (120:2) circle (2pt) (240:2) circle (2pt) node[below left] {$*$}; 
\draw(60:2) arc (60:-60:2) (120:2) arc (120:240:2);
\draw (60:2 and 1) arc (60:-240:2 and 1);
\fill(60:2 and 1) circle (2pt) (-240:2 and 1) circle (2pt) (-225:2 and 1) node[above=0.5pt ] {$e_2^v$} (-130:2 and 1) node[above] {$e_3^v$};
\begin{scope}[dashed,decoration={
    markings,
    mark=at position 0.5 with {\arrow{latex}}}
    ] 
    \draw[postaction={decorate}](60:2 and 1) arc (60:120:2 and 1)  (90:2 and 1) node[below] {$e_1$};
    \draw[postaction={decorate}](-60:2) arc (-60:-120:2) (-90:2) node[below] {$e_0$};
    \draw[postaction={decorate}](60:2) arc (60:120:2) (90:2) node[below] {$e_2$};
\end{scope}
\end{tikzpicture}
\end{align*}
\begin{proposition}\label{lemma:theta4presentation}
Group $\bfB_3(\Theta_4)$ is generated by the following generators
\begin{enumerate}
\item $Y$-exchanges $\sigma_1^{v;1,2},\sigma_1^{v;1,3},\sigma_1^{v;2,3}$,
\item one-particle loops $\gamma_1,\gamma_2$,
\item the circular move $\delta$.
\end{enumerate}
Moreover $\bfB_3(\Theta_4)$ has a presentation with the above generators and just one relation which reads
\begin{equation}\label{theta4-rel}
\sigma_1^{v;1,2}\gamma_1^{-1}\delta\sigma_1^{v;1,3}\delta^{-1}\gamma_1\sigma_1^{v;2,3}=\gamma_2^{-1}\delta\sigma_1^{v;2,3}\delta^{-1}\gamma_2\sigma_1^{v;1,3}\delta\sigma_1^{v;1,2}\delta^{-1}.
\end{equation}
\end{proposition}
\begin{proof}
Graph $\Theta_4$ is $2$-connected, hence by Proposition \ref{proposition:generators} group $\bfB_n(\Theta_4)$ is generated by $Y$-exchanges $\sigma_1^{v;a,b}$ and $\sigma_1^{w;a,b}$ for $1\leq a<b\leq 3$, one-particle moves $\gamma_1$, $\gamma_2$ and the circular move $\delta$. However, for $\Theta_4$, we actually may consider only $Y$-exchanges at $v$. To see this, note first that for each triple $(w,a,b)$ where $1\leq a<b\leq 3$ we have the corresponding triple $(v,c,d)$ such that edges $e^v_a,e^v_b,e^w_c,e^w_d$ belong to an image of an appropriate embedding of the $\Theta_3$-graph with root $*$ and spanning tree consistent with spanning the tree $T\subset \Theta_4$. Hence, any $Y$-exchange at vertex $w$ is $\Theta$-related to a $Y$-exchange at vertex $v$. 

Relation \eqref{theta4-rel} can be derived directly from the following pseudo-braiding relation \eqref{equation:pseudo-braid} for $\bfB_3(S_4)$
\[\sigma_1^{v;1,2}\sigma_2^{v;2,1,3}\sigma_1^{v;2,3}=\sigma_2^{v;1,2,3}\sigma_1^{v;1,3}\sigma_2^{v;3,1,2}.\]
To this end, associate with each $\sigma_2^{v;c,b,a}$ a lollipop graph containing $e_0$ in its loop. Then, lollipop relations as in \eqref{eq:lollipop-braid} can be applied. They are as follows.
\[\sigma_2^{v;1,2,3}=\gamma_2^{-1}\delta\sigma_1^{v;2,3}\delta^{-1}\gamma_2,\quad \sigma_2^{v;3,1,2}=\delta\sigma_1^{v;1,2}\delta^{-1},\quad \sigma_2^{v;2,1,3}=\gamma_1^{-1}\delta\sigma_1^{v;1,3}\delta^{-1}\gamma_1.\]
Relation \eqref{theta4-rel} is obtained by substituting the above lollipop relations in the pseudo-braid relation. Unfortunately, proving that relation \eqref{theta4-rel} is sufficient for presenting $\bfB_3(\Theta_4)$ requires referring to the Morse-theoretic methods. For the sake of completeness, we derive the corresponding Morse presentation of $\bfB_3(\Theta_4)$ in the Appendix.
\end{proof}

Let us next analyse a group presentation for $\bar\bfB_n(\Theta_4):=\bfB_n(\Theta_4)/\langle\gamma_0,\gamma_1, \gamma_{2}\rangle$. Firstly, recall that the one-particle move $\gamma_0$ in $C_n(\Theta_4)$ associated with deleted edge $e_0$ is defined as the move where i) top $(n-1)$ particles move to $e_1^v$ and ii) particle $n$ goes around the simple loop in $\Theta_4$ associated with edge $e_0$. Such a move satisfies lollipop relation \eqref{eq:delta-rel}, i.e.
\begin{equation}\label{theta4-lollipop}
\gamma_0=\sigma_{n-1}^{v;1,\hdots,1,1,3}\hdots\sigma_{2}^{v;1,1,3}\sigma_{1}^{v;1,3}\delta.
\end{equation}
Lollipop relation \eqref{theta4-lollipop} implies that in $\bar\bfB_n(\Theta_4)$ we have
\[\delta=\sigma_{1}^{v;3,1}\sigma_{2}^{v;1,3,1}\hdots\sigma_{n-1}^{v;1,\hdots,1,3,1}.\]
Furthermore, another set of lollipop relations as in \eqref{eq:lollipop-braid} implies that in $\bar\bfB_n(\Theta_4)$ we have 
\[\sigma_{i+1}^{v;a_1,a_2,\hdots,a_{i-1},a,b}=\delta \sigma_{i}^{v;a_2,\hdots,a_{i-1},a,b}\delta^{-1}\]
for any sequence $(a_1,a_2,\hdots,a_{i-1})$ such that $1\leq a_j\leq 3$. This in turn means that $Y$-exchanges $\sigma_i^{v;\bfa,a,b}$ in $\bar\bfB_n(\Theta_4)$ are in fact distinguished only by the choice of $1\leq a<b\leq 3$. Hence, we can drop decorations $\bfa$ and denote the respective equivalence class in $\bar\bfB_n(\Theta_4)$ as 
\[[\sigma_i^{v;\bfa,a,b}]:=\sigma_i^{v;a,b}=\delta^{i-1}\sigma_1^{v;a,b}\delta^{-i+1}.\]
With the above notation and identifications established, relation \eqref{theta4-rel} in $\bar\bfB_n(\Theta_4)$ becomes the pseudo-braiding relation
\begin{equation}\label{theta-yang-baxter}
\sigma_1^{v;1,2}\sigma_2^{v;1,3}\sigma_1^{v;2,3}=\sigma_2^{v;2,3}\sigma_1^{v;1,3}\sigma_2^{v;1,2}.
\end{equation}
Hence, when constructing unitary representations of $\bar\bfB_n(\Theta_4)$ on $\mH=\left(\CC^d\right)^{\otimes n}$ we need three $R$-matrices $R$, $R'$, $R''$ that separately constitute solutions of the Yang-Baxter equation and are assigned to generators of $\bar\bfB_n(\Theta_4)$ as follows
\begin{gather*}
\sigma_i^{v;1,2}\mapsto \bone\otimes\hdots\otimes\bone\otimes R_{i,i+1}\otimes\bone\otimes\hdots\otimes\bone,\\
\sigma_i^{v;1,3}\mapsto \bone\otimes\hdots\otimes\bone\otimes R'_{i,i+1}\otimes\bone\otimes\hdots\otimes\bone,\\
\sigma_i^{v;2,3}\mapsto \bone\otimes\hdots\otimes\bone\otimes R''_{i,i+1}\otimes\bone\otimes\hdots\otimes\bone.
\end{gather*}
On top of that, by \eqref{theta-yang-baxter} the $R$-matrices have to satisfy the following mixed Yang-Baxter equation
\begin{equation}\label{yang-baxter-mixed}
(R\otimes\bone)(\bone\otimes R')(R''\otimes\bone)=(\bone\otimes R'')(R'\otimes\bone)(\bone\otimes R).
\end{equation}

 \begin{acknowledgements}
The authors gratefully acknowledge the support of the American Institute of Mathematics (AIM) where this collaboration was initiated. We would like to thank Adam Sawicki and Jon Harrison for useful discussions during the workshop at AIM. TM wold like to also thank Jonathan Robbins and Nick Jones for helpful discussions about graph configuration spaces and anyons. Byung Hee An was supported by the National Research Foundation of Korea (NRF) grant funded by the Korea government (MSIT) No. 2020R1A2C1A01003201. TM acknowledges the support of the Foundation for Polish Science (FNP), START programme.
\end{acknowledgements}

\appendix
\section{Morse presentation of $\bfB_3(\Theta_4)$}\label{theta4-morse}
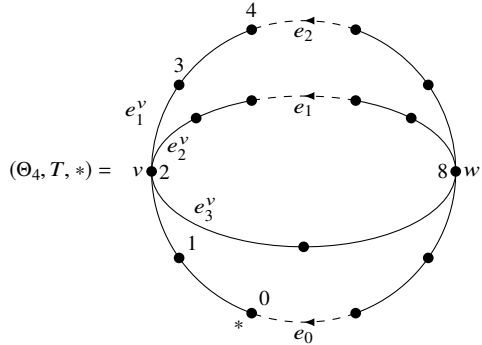
\begin{align*}
(\Theta_4,T,*)=\begin{tikzpicture}[baseline=-.5ex]
\fill(-2,0) circle (2pt) node[left] {$v$} node[right] {$2$} (2,0) circle (2pt) node[right] {$w$} node[left] {$8$} (70:2) circle (2pt) (165:2) node[above left] {$e_1^v$} (-70:2) circle (2pt) (110:2) circle (2pt) node[above=2pt] {$4$} (250:2) circle (2pt) node[below left] {$*$} node[above right] {$0$}; 
\fill(35:2) circle (2pt) (-35:2) circle (2pt) (145:2) circle (2pt) node[above=2pt] {$3$} (-145:2) circle (2pt) node[above right] {$1$};
\draw(70:2) arc (70:-70:2) (110:2) arc (110:250:2);
\draw (70:2 and 1) arc (70:-250:2 and 1);
\fill(70:2 and 1) circle (2pt) (-250:2 and 1) circle (2pt) (-215:2 and 1) node[below=0.5pt ] {$e_2^v$} (-130:2 and 1) node[above] {$e_3^v$} (-90:2 and 1) circle (2pt) (135:2 and 1) circle (2pt) (45:2 and 1) circle (2pt);
\begin{scope}[dashed,decoration={
    markings,
    mark=at position 0.5 with {\arrow{latex}}}
    ] 
    \draw[postaction={decorate}](70:2 and 1) arc (70:110:2 and 1) (90:2 and 1) node[below] {$e_1$};
    \draw[postaction={decorate}](-70:2) arc (-70:-110:2) (-90:2) node[below] {$e_0$};
    \draw[postaction={decorate}](70:2) arc (70:110:2) (90:2) node[below] {$e_2$};
\end{scope}
\end{tikzpicture}
\end{align*}
We will be using only some particular cells of Morse complex $M_3(\Theta_4)$, but for the sake of completeness we write down all critical $1$-cells of $M_3(\Theta_4)$ and all critical $2$-cells together with their boundary words. They were found using one of the author's own computer code implementing Farley-Sabalka's algorithm \cite{TMcode}. Critical $1$-cells read:
\begin{gather*}
g_{0}:=\{e_{2}^{7},3,8\},\quad g_{1}:=\{e_{2}^{5},3,7\},\quad g_{2}:=\{e_{8}^{11},0,9\},\quad g_{3}:=\{e_{8}^{13},11,12\},\quad g_{4}:=\{e_{8}^{13},9,11\}, \\ 
g_{5}:=\{e_{2}^{5},3,4\},\quad g_{6}:=\{e_{4}^{12},0,1\}, \quad g_{7}:=\{e_{8}^{13},9,14\}, \quad g_{8}:=\{e_{2}^{5},3,6\},\quad g_{9}:=\{e_{8}^{13},0,11\}, \\
g_{10}:=\{e_{8}^{11},9,13\},\quad g_{11}:=\{e_{8}^{11},9,12\},\quad g_{12}:=\{e_{8}^{13},9,10\},\quad g_{13}:=\{e_{8}^{13},11,14\}, \\
g_{14}:=\{e_{8}^{11},9,10\}, \quad  g_{15}:=\{e_{2}^{7},3,5\},\quad g_{16}:=\{e_{2}^{7},5,8\},\quad g_{17}:=\{e_{8}^{13},0,9\},\quad g_{18}:=\{e_{2}^{5},0,3\},\\
g_{19}:=\{e_{0}^{14},1,2\}, \quad g_{20}:=\{e_{2}^{7},0,3\},\quad g_{21}:=\{e_{2}^{7},0,5\},\quad g_{22}:=\{e_{6}^{10},0,1\},\quad g_{23}:=\{e_{2}^{7},5,6\}, \\
g_{24}:=\{e_{2}^{7},3,4\}.
\end{gather*}
Critical $2$-cells and their boundary words read
\begin{gather*}
\partial\{e_{0}^{14},e_{2}^{7},6\}=g_{19}g_{21}^{-1}g_{19}^{-1}g_{16},\quad \partial\{e_{2}^{7},e_{4}^{12},6\}=g_{15}g_{18}g_{6}^{-1}g_{16}^{-1}g_{6}g_{20}^{-1}g_{1}^{-1}, \\
\partial\{e_{6}^{10},e_{8}^{13},10\}=g_{22}g_{21}^{-1}g_{16}^{-1}g_{17}^{-1}g_{16}g_{21}g_{22}^{-1}g_{17}^{-1}g_{12},\quad \partial\{e_{4}^{12},e_{8}^{13},10\}=g_{6}g_{20}^{-1}g_{0}^{-1}g_{17}^{-1}g_{0}g_{20}g_{6}^{-1}g_{9}^{-1}g_{4}, \\
\partial\{e_{0}^{14},e_{6}^{10},2\}=g_{19}g_{22}^{-1}g_{19}^{-1}g_{17}g_{22}g_{21}^{-1}, \quad \partial\{e_{0}^{14},e_{8}^{13},10\}=g_{19}g_{17}^{-1}g_{19}^{-1}g_{7}, \\
\partial\{e_{4}^{12},e_{8}^{13},12\}=g_{6}g_{20}^{-1}g_{0}^{-1}g_{9}^{-1}g_{0}g_{20}g_{6}^{-1}g_{9}^{-1}g_{3},\quad \partial\{e_{2}^{7},e_{4}^{12},4\}=g_{24}g_{6}^{-1}g_{0}^{-1}g_{6}g_{20}^{-1}, \\
\partial\{e_{2}^{7},e_{6}^{10},6\}=g_{23}g_{22}^{-1}g_{16}^{-1}g_{22}g_{21}^{-1}, \quad \partial\{e_{0}^{14},e_{8}^{11},10\}=g_{19}g_{2}^{-1}g_{19}^{-1}g_{10}, \\
\partial\{e_{6}^{10},e_{8}^{13},12\}=g_{2}g_{22}g_{21}^{-1}g_{16}^{-1}g_{9}^{-1}g_{16}g_{21}g_{22}^{-1}g_{17}^{-1}g_{10}^{-1}g_{4},\quad \partial\{e_{2}^{5},e_{4}^{12},4\}=g_{5}g_{6}^{-1}g_{1}^{-1}g_{6}g_{18}^{-1}, \\
\partial\{e_{4}^{12},e_{8}^{11},10\}=g_{6}g_{20}^{-1}g_{0}^{-1}g_{2}^{-1}g_{0}g_{20}g_{6}^{-1}g_{11}, \quad \partial\{e_{0}^{14},e_{2}^{7},4\}=g_{19}g_{20}^{-1}g_{19}^{-1}g_{0}, \\
\partial\{e_{6}^{10},e_{8}^{11},10\}=g_{22}g_{21}^{-1}g_{16}^{-1}g_{2}^{-1}g_{16}g_{21}g_{22}^{-1}g_{2}^{-1}g_{14},\quad \partial\{e_{4}^{12},e_{6}^{10},1\}=g_{6}g_{18}^{-1}g_{22}^{-1}g_{20}g_{6}^{-1}g_{2}g_{22}g_{21}^{-1}, \\
\partial\{e_{0}^{14},e_{8}^{13},12\}=g_{19}g_{9}^{-1}g_{19}^{-1}g_{13}, \quad \partial\{e_{2}^{5},e_{6}^{10},4\}=g_{8}g_{22}^{-1}g_{1}^{-1}g_{22}, \\
\partial\{e_{0}^{14},e_{2}^{5},4\}=g_{19}g_{18}^{-1}g_{19}^{-1}g_{1}, \quad \partial\{e_{2}^{7},e_{6}^{10},4\}=g_{15}g_{22}^{-1}g_{0}^{-1}g_{22}g_{21}^{-1}, \\
\partial\{e_{0}^{14},e_{4}^{12},2\}=g_{19}g_{6}^{-1}g_{19}^{-1}g_{9}g_{6}g_{20}^{-1}.
\end{gather*}
The above Morse presentation uses $25$ generators and $21$ relators, but via appropriate Tietze transformations it can be reduced to the following presentation on $6$ generators and one relator.
\begin{equation}\label{theta-morse}
\bfB_3(\Theta_4)=\left\langle g_{18},g_{20},g_{21},g_{22},g_6,g_{19}\mid\ g_{21}g_{22}^{-1}g_{19}g_{20}g_{19}^{-1}g_{22}g_{18}g_6^{-1}g_{19}g_{21}^{-1}g_{19}^{-1}g_6g_{20}^{-1}g_{19}g_{18}^{-1}g_{19}^{-1}\right\rangle.
\end{equation}
The only relator in the above presentation is derived from $\partial\{e_{6}^{10},e_{8}^{13},12\}$ by substituting expressions for $g_2$, $g_{16}$, $g_{17}$, $g_{10}$, $g_{4}$, $g_9$. Expression for $g_2$ is obtained directly from $\partial\{e_{4}^{12},e_{6}^{10},1\}$. Expression for $g_{16}$ is obtained directly from $\partial\{e_{0}^{14},e_{2}^{7},6\}$. Expression for $g_{17}$ is obtained directly from $\partial\{e_{0}^{14},e_{6}^{10},2\}$. Expression for $g_{10}$ is obtained from $\partial\{e_{0}^{14},e_{8}^{11},10\}$ by substituting expression for $g_2$ which is extracted from $\partial\{e_{4}^{12},e_{6}^{10},1\}$. Expression for $g_9$ is obtained directly from $\partial\{e_{0}^{14},e_{4}^{12},2\}$. Finally, $g_4$ can be extracted from $\partial\{e_{4}^{12},e_{8}^{13},10\}$ using previous results and the expression for $g_0$ derived from $\partial\{e_{0}^{14},e_{2}^{7},4\}$.

Let us emphasise that although presentation \eqref{theta-morse} was derived for $n=3$, the relator holds for any $n$. This is because all relators from Morse presentation for $n$ particles can be extended to relators of the Morse presentation for $n+1$ particles by subdividing $\graf$ and adding one more particle next to the root $*$ to all relevant critical cells \cite{TMpresentations}. 

Let us next interpret presentation \eqref{theta-morse} in terms of geometric generators that correspond to particle moves on $\Theta_4$. We have the following correspondence
\[g_{18}\sim\sigma_1^{v;2,1},\quad g_{20}\sim\sigma_1^{v;3,1},\quad g_{21}\sim\sigma_1^{v;3,2},\quad g_{22}\sim \gamma_1,\quad, g_6\sim\gamma_2,\quad g_{19}\sim\delta.\]
The above correspondence follows from the correspondence between loop-generators and $Y$-exchanges and critical cells established in the proof of Theorem \ref{thm:generators} and its preceding sections. Namely, critical cells $g_{18}$, $g_{20}$ and $g_{21}$ correspond to essential vertex $v$ with label $2$. Hence, apply formula \eqref{equation:critical cell interpretation} to find their corresponding $Y$-exchanges. Furthermore, critical cells $g_{22}$, $g_{6}$ and $g_{19}$ are associated with deleted edges, hence they correspond to loop-type generators and the circular move. Cell $g_{19}$ corresponds to deleted edge $e_0$ which is incident at the root $*$, hence it corresponds to the circular move $\delta$. This way, we translate relator from \eqref{theta-morse} to the following relation involving geometric generators.
\[\sigma_1^{v;3,2}\gamma_1^{-1}\delta\sigma_1^{v;3,1}\delta^{-1}\gamma_1\sigma_1^{v;2,1}\gamma_2^{-1}=\delta\sigma_1^{v;2,1}\delta^{-1}\sigma_1^{v;3,1}\gamma_2^{-1}\delta\sigma_1^{v;3,2}\delta^{-1}.\]

\end{document}